\documentclass[11pt]{article}
\def\withcolors{0}
\def\withnotes{0}
 \makeatletter{}\usepackage[T1]{fontenc}
\usepackage[utf8]{inputenc}

\usepackage{lmodern}
\usepackage{xspace}
\usepackage[normalem]{ulem}

\usepackage{amsfonts,amsmath,amssymb, amsthm, mathtools}
\usepackage{thm-restate}
\usepackage{dsfont} 
\usepackage{algorithmicx,algpseudocode,algorithm}

\usepackage[usenames,dvipsnames,table]{xcolor}

\usepackage{csquotes}

\usepackage{relsize}

\usepackage{multirow}
\usepackage{chngpage} 
\usepackage{mfirstuc}

\usepackage{tikz}
\usetikzlibrary{arrows}
\usetikzlibrary{calc,decorations.pathmorphing,patterns}

\usepackage[backref,colorlinks,citecolor=blue,bookmarks=true,linktocpage]{hyperref}
\usepackage{aliascnt}
\usepackage[numbered]{bookmark}

\usepackage{accents}

\usepackage{fullpage}

\usepackage{titling}

\usepackage[shortlabels]{enumitem}
  \setitemize{noitemsep,topsep=3pt,parsep=2pt,partopsep=2pt}   \setenumerate{itemsep=1pt,topsep=2pt,parsep=2pt,partopsep=2pt}
  \setdescription{itemsep=1pt}
  
\ifnum\withnotes=1
  \usepackage[colorinlistoftodos,textsize=scriptsize]{todonotes}
\fi

\newcommand{\marginnote}[1]{}

\usepackage{mleftright}  
\makeatletter{}\makeatletter
\@ifundefined{theorem}{    \theoremstyle{plain}   	\newtheorem{theorem}{Theorem}[section]
  	\newaliascnt{coro}{theorem}
  	  \newtheorem{corollary}[coro]{Corollary}
  	\aliascntresetthe{coro}
  	\newaliascnt{lem}{theorem}
  		\newtheorem{lemma}[lem]{Lemma}
  	\aliascntresetthe{lem}
  	\newaliascnt{clm}{theorem}
  		\newtheorem{claim}[clm]{Claim}
	\aliascntresetthe{clm}
	\newaliascnt{fact}{theorem}
 	 	\newtheorem{fact}[theorem]{Fact}
	\aliascntresetthe{fact}
  	
  \newaliascnt{prop}{theorem}
  		\newtheorem{proposition}[prop]{Proposition}
	\aliascntresetthe{prop}
	\newaliascnt{conj}{theorem}
  		
	\aliascntresetthe{conj}
  \theoremstyle{remark}

  \theoremstyle{definition}   	\newaliascnt{defn}{theorem}
 		 \newtheorem{definition}[defn]{Definition}
 	 \aliascntresetthe{defn}
}{}
\makeatother

\newenvironment{proofof}[1]{\begin{proof}[Proof of {#1}]}{\end{proof}}

\providecommand{\email}[1]{\href{mailto:#1}{\nolinkurl{#1}\xspace}}

\ifnum\withcolors=1
  \newcommand{\new}[1]{{\color{orange} {#1}}}      \newcommand{\newest}[1]{{\color{red} {#1}}}      \newcommand{\acolor}[1]{{\color{orange}#1}}   \newcommand{\ccolor}[1]{{\color{black!40!cyan}#1}}   \newcommand{\ecolor}[1]{{\color{black!50!green}#1}}   \newcommand{\kcolor}[1]{{\color{Plum}#1}}   \newcommand{\scolor}[1]{{\color{RubineRed}#1}} \else
  \newcommand{\new}[1]{{{#1}}}
  
  \newcommand{\newest}[1]{{{#1}}}
  
  \newcommand{\acolor}[1]{{#1}}
  \newcommand{\ccolor}[1]{{#1}}
  \newcommand{\ecolor}[1]{{#1}}
  \newcommand{\kcolor}[1]{{#1}}
  \newcommand{\scolor}[1]{{#1}}
\fi
\newcommand{\eps}{\ensuremath{\varepsilon}\xspace}
\newcommand{\Algo}{\ensuremath{\mathcal{A}}\xspace} \newcommand{\Tester}{\ensuremath{\mathcal{T}}\xspace}  \newcommand{\property}{\ensuremath{\mathcal{P}}\xspace} \newcommand{\class}{\ensuremath{\mathcal{C}}\xspace} \newcommand{\eqdef}{\stackrel{\rm def}{=}}

\newcommand{\accept}{\textsf{ACCEPT}\xspace}

\newcommand{\reject}{\textsf{REJECT}\xspace}
\newcommand{\opt}{{\textsc{opt}}\xspace}

\newcommand{\domain}{\ensuremath{\Omega}\xspace}

\newcommand{\littleO}[1]{{o\mleft( #1 \mright)}}
\newcommand{\bigO}[1]{{O\mleft( #1 \mright)}}

\newcommand{\bigTheta}[1]{{\Theta\mleft( #1 \mright)}}
\newcommand{\bigOmega}[1]{{\Omega\mleft( #1 \mright)}}

\newcommand{\tildeO}[1]{\tilde{O}\mleft( #1 \mright)}

\newcommand{\tildeOmega}[1]{\operatorname{\tilde{\Omega}}\mleft( #1 \mright)}
\providecommand{\poly}{\operatorname*{poly}}

\newcommand{\totinf}[1][f]{{\mathbf{Inf}[#1]}}
\newcommand{\infl}[2][f]{{\mathbf{Inf}_{#2}[#1]}}

\newcommand{\setOfSuchThat}[2]{ \left\{\; #1 \;\colon\; #2\; \right\} } 			\newcommand{\indicSet}[1]{\mathds{1}_{#1}}                                              \newcommand{\indic}[1]{\indicSet{\left\{#1\right\}}}                                             

\newcommand{\dist}[2]{\operatorname{dist}\!\left(#1, #2\right)}

\newcommand{\lpdist}[3][p]{\lp[#1]\!\left(#2, #3\right)}

\newcommand\restr[2]{{  \left.\kern-\nulldelimiterspace   #1   \vphantom{\big|}   \right|_{#2}   }}

\newcommand{\proba}{\Pr}
\newcommand{\probaOf}[1]{\proba\!\left[\, #1\, \right]}
\newcommand{\probaCond}[2]{\proba\!\left[\, #1 \;\middle\vert\; #2\, \right]}
\newcommand{\probaDistrOf}[2]{\proba_{#1}\left[\, #2\, \right]}

\newcommand{\supp}[1]{\operatorname{supp}\!\left(#1\right)}

\newcommand{\shortexpect}{\mathbb{E}}

\newcommand{\poisson}[1]{\ensuremath{\operatorname{Poisson}\!\left( #1 \right) }}

\newcommand{\norm}[1]{\lVert#1{\rVert}}

\newcommand{\normtwo}[1]{{\norm{#1}}_2}

\newcommand{\abs}[1]{\left\lvert #1 \right\rvert}

\newcommand{\dotprod}[2]{ \left\langle #1,\xspace #2 \right\rangle } 			 			
 			
\newcommand{\clg}[1]{\left\lceil #1 \right\rceil}
\newcommand{\flr}[1]{\left\lfloor #1 \right\rfloor}

\newcommand{\R}{\ensuremath{\mathbb{R}}\xspace}

\newcommand{\N}{\ensuremath{\mathbb{N}}\xspace}

\def\B{\{0,1\}}

\newcommand{\pdfsamp}{dual\xspace}
\newcommand{\cdfsamp}{cumulative dual\xspace}
\newcommand{\Pdfsamp}{\expandafter\capitalisewords\expandafter{\pdfsamp}}
\newcommand{\Cdfsamp}{\expandafter\capitalisewords\expandafter{\cdfsamp}}

\newcommand{\lp}[1][1]{L_{#1}}

\newcommand{\D}{\ensuremath{D}}

\newcommand{\ubar}[1]{\underaccent{\bar}{#1}}
\newcommand{\hseq}[1]{\bar{\partial}{#1}}
\newcommand{\lseq}[1]{\ubar{\partial}{#1}}
\newcommand{\unknown}{\ast}
\newcommand{\kmon}[2]{\mathcal{M}^{(#1)}_{#2}}

\makeatletter

\newcommand{\Rom}[1]{\expandafter\@slowromancap\romannumeral #1@}

\makeatother

\makeatletter
\newcommand\ackname{Acknowledgements}
\if@titlepage
  \newenvironment{acknowledgements}{      \titlepage
      \null\vfil
      \@beginparpenalty\@lowpenalty
      \begin{center}        \bfseries \ackname
        \@endparpenalty\@M
      \end{center}}     {\par\vfil\null\endtitlepage}
\else
  
\fi
\makeatother

\def\authornamecc{Cl\'ement L. Canonne}
\def\authorafficc{Columbia University. Email: \email{ccanonne@cs.columbia.edu}. Research supported by NSF CCF-1115703 and NSF CCF-1319788.}
\def\authornameeg{Elena Grigorescu}
\def\authoraffieg{Purdue University. Email: \email{elena-g@purdue.edu}. Research supported in part by NSF CCF-1649515.}
\def\authornamesg{Siyao Guo}
\def\authoraffisg{Courant Institute of Mathematical Sciences, New York University. Email: \email{sg191@nyu.edu}.}
\def\authornameak{Akash Kumar}
\def\authoraffiak{Purdue University. Email: \email{akumar@purdue.edu}.}
\def\authornamekw{Karl Wimmer}
\def\authoraffikw{Duquesne University. Email: \email{wimmerk@duq.edu}.}

\title{Testing $k$-Monotonicity\\ {\normalsize (The Rise and Fall of Boolean Functions)}}
\date{\today}
 
\author{
  \ccolor{\authornamecc}\thanks{\authorafficc}
  \and \ecolor{\authornameeg}\thanks{\authoraffieg}
  \and \scolor{\authornamesg}\thanks{\authoraffisg}
  \and \acolor{\authornameak}\thanks{\authoraffiak}
  \and \kcolor{\authornamekw}\thanks{\authoraffikw}
}

\begin{document}

\maketitle

\begin{abstract}

 A Boolean {\em  $k$-monotone} function defined over a finite poset domain ${\cal D}$ 
alternates between the values $0$ and $1$ at most $k$ times on any ascending chain in ${\cal D}$. Therefore, $k$-monotone functions are natural generalizations of the classical {\em monotone} functions,  which are the {\em $1$-monotone} functions.

Motivated by the recent interest in $k$-monotone functions in the context of circuit complexity and learning theory, and by the central role that  monotonicity testing plays in the context of property testing, we initiate a systematic study of   $k$-monotone functions, in the property testing model. In this model, the goal is to distinguish functions that are $k$-monotone (or are close to being $k$-monotone) from functions that are far from being $k$-monotone.

Our results include the following:

\begin{enumerate}
\item We demonstrate a separation between testing $k$-monotonicity and testing monotonicity, on the hypercube domain $\{0,1\}^d$, for $k\geq 3$;
\item We demonstrate a separation between testing and learning  on $\{0,1\}^d$, for $k=\omega(\log d)$: testing $k$-monotonicity  can be performed with $2^{O(\sqrt d \cdot \log d\cdot \log{1/\eps})}$ queries,  while learning $k$-monotone functions requires $2^{\Omega(k\cdot \sqrt d\cdot{1/\eps})}$ queries (Blais et al. (RANDOM 2015)).
\item We present a tolerant test for functions $f\colon[n]^d\to \{0,1\}$ with complexity independent of $n$, which makes progress on a problem left open by Berman et al. (STOC 2014).
\end{enumerate}

 Our techniques exploit the testing-by-learning paradigm, use novel applications of  Fourier analysis on the grid $[n]^d$, and draw connections to distribution testing techniques.

\end{abstract}

\clearpage
\tableofcontents
\clearpage

\section{Introduction}\label{sec:intro}
\makeatletter{}\newcommand{\ra}{\rightarrow}
\newcommand{\etal}{{\em{et al}}}

A function $f\colon {\cal D} \ra \{0,1\}$, defined over a finite domain $\cal{D}$ equipped with a partial order, is said to be $k$-{\em monotone}, for some integer $k\geq 0$, if there does not exist $x_1\preceq x_2 \preceq \ldots \preceq x_{k+1}$ in ${\cal D}$ such that $f(x_1)=1$ and $f(x_i)\ne f(x_{i+1})$ for all $i\in [k]$. Note that $1$-monotone functions are the classical {\em monotone} functions, satisfying $f(x_1)\leq f(x_2)$, whenever $x_1\preceq x_2$.

Monotone functions  have been well-studied on multiple fronts in computational complexity due to their natural structure. They have been celebrated for decades  in the  property testing  literature~\cite{GGLRS00, DGLRRS99, FLNRRS:02, BrietCGM10, CS:13a, CS:13b, CS:13}, where we have recently witnessed ultimate results~\cite{KhotMS:15,ChenDST:15,BelovsB:15}, in the circuit complexity literature, where we now have strong lower bounds~\cite{RazW:92, Razborov:85-russianjournal}, and in computational learning, where we now have learning algorithms in numerous learning models~\cite{BshoutyT96, Angluin87,KearnsV94, Servedio04, ODonnellS07, ODonnellW09}.

The generalized notion of $k$-monotonicity has also  been studied in the context of circuit lower bounds for more than 50 years. In particular,  Markov~\cite{Markov:57} showed that any $k$-monotone function (even with multiple outputs) can be computed using circuits containing only $\log{k}$ negation gates.  The presence of negation gates appears to be a challenge in proving circuit lower bounds: ``the effect of such gates on circuit size remains to a large extent a mystery'' \cite{Jukna:12}.
 The recent results of Blais \etal.~\cite{BlaisCOST:15} on circuit lower bounds 
  have prompted renewed interest in understanding $k$-monotone functions from multiple angles,  including cryptography,  circuit complexity, learning theory, and Fourier analysis (\cite{Rossman:15, GuoMOR:15, GuoK:15, LinZ:16}).

Motivated by the exponential lower bounds on PAC learning $k$-monotone functions due to~\cite{BlaisCOST:15}, we initiate the study of $k$-monotonicity in the closely related  {\em Property Testing} model.  In this model, given query access to a function, one must decide if the function is $k$-monotone, or is far from being $k$-monotone, by querying the input only in a small number of places.

\makeatletter{}
\subsection{Our results}

 We focus on testing $k$-monotonicity of Boolean functions defined over the $d$-dimensional hypegrid $[n]^d$, and the  hypercube 
 $\{0,1\}^d$.    We begin our presentation with the results for the hypercube, in order to build  intuition into the difficulty of the problem, while comparing our results with the current literature on testing monotonicity. Our stronger results concern the hypegrid $[n]^d$.
    
\subsubsection{Testing $k$-monotonicity on the hypercube $\{0,1\}^d$}

In light of the recent results of \cite{KhotMS:15} that provide a $O(\sqrt{d})$-query tester for monotonicity, we first show that testing $k$-monotonicity is strictly harder than testing monotonicity on $\{0,1\}^d$, for $k\geq 3$.

\begin{restatable}{theorem}{onesidedhypercube}\label{thm:1-sided-non-adaptive-lb}
	For $1\leq k\leq d^{1/4}/2$, any one-sided non-adaptive tester for $k$-monotonicity of functions $f\colon\{0, 1\}^d\to \{0,1\}$ must make $(\bigOmega{d/k^2})^{k/4}$ queries.
\end{restatable}

Both \autoref{thm:1-sided-non-adaptive-lb} and its proof generalize the $\Omega(d^{1/2})$ lower bound for testing monotonicity, due to  Fischer\ \etal.~\cite{FLNRRS:02}.

On the upper bounds side, while the monotonicity testing problem is providing numerous potential techniques  for approaching this new problem \cite{GGLRS00, DGLRRS99, CS:13b, BrietCGM10, ChenST:14, KhotMS:15}, most common techniques  appear to resist generalizations to $k$-monotonicity.  However, our upper bounds  demonstrate a  separation between testing and PAC  learning $k$-monotonicity, for large enough values of $k=\omega(\log d)$.

\begin{restatable}{theorem}{testcubeonesidednaub}\label{theo:test:d:os:na:ub}
There exists a one-sided non-adaptive tester for $k$-monotonicity of functions $f\colon \B^d\to\{0,1\}$ with query complexity $q(d,\eps,k)=2^{\bigO{\sqrt{d}\cdot \log d \cdot \log\frac{1}{\eps}}}$.
\end{restatable}

Indeed, in the related PAC learning model, \cite{BlaisCOST:15} shows that learning $k$-monotone functions on the hypercube requires $2^{\Omega(k\cdot \sqrt d\cdot 1/\eps)}$ many queries.

We further observe that the recent non-adaptive and adaptive $2$-sided lower bounds of \cite{ChenDST:15, BelovsB:15}, imply the same bounds for $k$-monotonicity, using black box reductions. We summarize the state of the art for testing $k$-monotonicity on the hypercube in \autoref{table:cube-boolean-k-results}.

\begin{table}[H]\centering
  \begin{tabular}{|l|c|c|c|c|}
  \hline
    & upper bound &  1.s.-n.a. lower bound  & 2.s.-n.a. lower bound & 2.s.-a. lower bound\\
         \hline
   $k=1$ & $\bigO{\sqrt{d}}$~\cite{KhotMS:15} &  $\bigOmega{d^{1/2}}$~\cite{FLNRRS:02}
   &  $\bigOmega{d^{1/2-o(1)}}$~\cite{ChenDST:15}   &  $\tildeOmega{d^{1/4}}$~\cite{BelovsB:15}\\
   $k\geq 2$ & \parbox[t]{3cm}{ $d^{\bigO{k\sqrt{d}}}$~\cite{BlaisCOST:15},\\ $d^{\bigO{\sqrt{d}}}$ Thm~\ref{theo:test:d:os:na:ub} \strut} &   \parbox[t]{3cm}{  $\bigOmega{d/k^2}^{k/4}$~Thm~\ref{thm:1-sided-non-adaptive-lb} \\ ($k=O(d^{1/4})$)}
   &  $\bigOmega{d^{1/2-o(1)}}$ Cor~\ref{theo:k-monotone-carry} &  $\tildeOmega{d^{1/4}}$ Cor~\ref{theo:k-monotone-carry}\\
  \hline
  \end{tabular}
  \caption{\label{table:cube-boolean-k-results} Testing $k$-monotonicity of a function $f\colon\{0,1\}^d\to\{0,1\}$}
\end{table}

\subsubsection{Testing $k$-monotonicity on the hypergrid $[n]^d$}

The remainder of the paper  focuses on functions defined over the $d$-dimensional hypergrid domain $[n]^d$, where we denote by $(i_1, i_2, \ldots, i_n)\preceq (j_1, j_2, \ldots, j_n)$ the partial order in which $i_1\leq j_1, i_2\leq j_2, \ldots, i_n\leq j_n$.  Testing monotonicity has received a lot of attention over the $d$-dimensional hypergrids \cite{GGLRS00, ErgunKKRV00, Fischer04,  BatuRW05, AilonC06, HalevyK08,  BlaisBM12, CS:13a, CS:13b,  CS:13,  BermanRY:14}, where the problem is well-understood, and we refer the reader to~\autoref{table:grids-boolean} in the appendix for a detailed review on the state of the art in the area.
We summarize our results on testing $k$-monotonicity over $[n]^d$ in \autoref{table:results}.

\begin{table}[H]\centering
  \begin{tabular}{|l|c|c|c|}
  \hline
         & General $k$ &  $k=2$ & $k=1$ (monotonicity) \\
         \hline
   $d=1$ & $\bigTheta{\frac{k}{\eps}}$ 1.s.-n.a., $\tildeO{\frac{1}{\eps^7}}$ 2.s.-n.a. & $\bigO{\frac{1}{\eps}}$ 1.s.-n.a. & { $\bigTheta{\frac{1}{\eps}}$ 1.s.-n.a. } \\\hline
   $d=2$ & $\tildeO{\frac{k^2}{\eps^3}}$  2.s.-n.a. (from below) & $\bigTheta{\frac{1}{\eps} }$ 2.s.-a. 
   & { $\bigTheta{\frac{1}{\eps}\log\frac{1}{\eps}}$ 1.s.-n.a., $\bigTheta{\frac{1}{\eps}}$ 1.s.-a. } \\\hline
\multirow{2}{*}{$d\geq 3$} & $\tildeO{ \frac{1}{\eps^2} \left(\frac{5kd}{\eps}\right)^d }$ 2.s.-n.a.,& $\tildeO{ \frac{1}{\eps^2}\left(\frac{10d}{\eps}\right)^d }$  2.s.-n.a. & \multirow{2}{*}{ $\bigO{\frac{d}{\eps}\log\frac{d}{\eps}}$ 1.s.-n.a. }  \\
& $2^{\tildeO{k\sqrt{d}/\eps^2}}$ 2.s.-n.a. & $2^{\tildeO{\sqrt{d}/\eps^2}}$  2.s.-n.a. &  \\
  \hline
  \end{tabular}
  \caption{\label{table:results}Summary of our results: testing $k$-monotonicity of a function $f\colon[n]^d \to \{0,1\}$ (first two columns). The last column contains known bounds on monotonicity testing (see~\autoref{app:overview}) and is provided for comparison.}
\end{table}

\paragraph{Testing $k$-monotonicity on the line and the 2-dimensional grid} We begin with a study of function $f\colon[n]\ra \{0,1\}$.
As before, note that $1$-sided tests should always accept $k$-monotone functions, and so, they must accept unless they discover a violation to $k$-monotonicity in the form of a  sequence  $x_1\preceq x_2 \preceq \ldots \preceq x_{k+1}$ in $[n]^d$, such that $f(x_1)=1$ and $f(x_i)\ne f(x_{i+1})$. Therefore, lower bounds for $1$-sided $k$-monotonicity testing must grow at least linearly with $k$. We show that this is indeed the case for both adaptive and non-adaptive tests, and more over, we give a tight non-adaptive algorithm. Consequently, our results demonstrate that adaptivity does not help in testing $k$-monotonicity with one-sided error on the line domain.
 
\begin{restatable}{theorem}{testlineonesidednalb}\label{theo:test:d1:os:a:lb}
Any one-sided (possibly adaptive) tester for $k$-monotonicity of functions $f\colon [n]\to\{0,1\}$ must have query complexity $\bigOmega{\frac{k}{\eps}}$.
\end{restatable}

The upper bound generalizes the $O(1/\eps)$ tester for monotonicity on the line.

\begin{restatable}{theorem}{testlineonesidednaub}\label{theo:test:d1:os:na:ub}
There exists a one-sided non-adaptive tester for $k$-monotonicity of functions $f\colon [n]\to\{0,1\}$ with query complexity $q(n,\eps,k)=\bigO{\frac{k}{\eps}}$.
\end{restatable}

Testing with $2$-sided error, however, does not require a dependence on $k$. In fact the problem has been well-studied in the machine learning literature in the context of testing/learning ``union of intervals''~\cite{KearnsR:00,BalcanBBY:12}, and in testing geometric properties, in the context of testing surface area~\cite{KothariNOW:14,Neeman:14},\footnote{We thank Eric Blais for mentioning the connection, and pointing us to these works} resulting in an $O({1/\eps}^{{7}/{2}})$-query algorithm. Namely, the starting point of~\cite{BalcanBBY:12} (later improved by~\cite{KothariNOW:14}) is a ``Buffon Needle's''-type argument, where the crucial quantity to analyze is the noise sensitivity of the function that is the probability that a randomly chosen pair of nearby points cross a ``boundary'' -- i.e., have different values. (Moreover, the algorithm of~\cite{BalcanBBY:12} works in the \emph{active testing} setting: it only requires a weaker access model that the standard query model.)

We provide an alternate proof of a $\poly(1/\eps)$ bound (albeit with a worse exponent)  that reveals a surprising connection with {\em  distribution testing}, namely with the problem of estimating the support size of  a distribution.

\begin{restatable}{theorem}{testlinetwosidednaub}\label{theo:test:d1:2s:na:ub}
There exists a two-sided non-adaptive tester for $k$-monotonicity of functions $f\colon [n]\to\{0,1\}$ with query complexity $q(n,\eps,k)=\tildeO{{1}/{\eps^7}}$, \emph{independent of $k$}.
\end{restatable}

An immediate implication of~\autoref{theo:test:d1:2s:na:ub} is that one can test even $n^{1-\alpha}$-monotonicity of $f\colon[n]\ra \{0,1\}$, for every $\alpha>0$, with a constant number of queries. Hence,  there is a separation between $1$-sided and $2$-sided testing, for $k=\omega(1)$.

Turning to the $2$-dimensional grid, we show that $2$-monotone functions can be tested with the minimum number of queries one could hope for:
\begin{restatable}{theorem}{testgridtwosidedaub}\label{theo:test:d2:2s:a:ub}
There exists a two-sided adaptive tester for $2$-monotonicity of functions $f\colon [n]^2\to\{0,1\}$ with query complexity $q(n,\eps)=\bigO{\frac{1}{\eps}}$.
\end{restatable}

\newest{We also discuss possible generalizations of~\autoref{theo:test:d2:2s:a:ub} to general $k$ or $d$  section,~\autoref{ssec:grid:kk}.}

\paragraph{Testing $k$-monotonicity on $[n]^d$, tolerant testing, and distance approximation}

Moving to the general grid domain  $[n]^d$, we show that $k$-monotonicity is testable with $\poly(1/\eps, k)$ queries  in constant-dimension grids.

\begin{restatable}{theorem}{testhighdim}\label{coro:tol:test:d:na}
There exists a non-adaptive tester for $k$-monotonicity of functions $f\colon [n]^d\to\{0,1\}$ with query complexity 
$q(n,d,\eps,k)=\min(\tildeO{ \frac{1}{\eps^2} \left(\frac{5kd}{\eps}\right)^d }, 2^{\tildeO{k\sqrt{d}/\eps^2}})$.
\end{restatable}

In fact, we obtain more general testing algorithms than in~\autoref{coro:tol:test:d:na}, namely our results hold for {\em tolerant} testers.

 The notion of tolerant testing was first introduced in~\cite{PRR:06} to account for the possibility of noisy data. In this notion, a test should  accept inputs that are $\eps_1$-close to the property, and reject inputs that are $\eps_2$-far from the property, where $\eps_1$ and $\eps_2$ are given parameters. Tolerant testing is intimately connected to the notion of distance approximation:  given tolerant testers for every $(\eps_1, \eps_2)$, there exists an algorithm that estimates the distance to the property within any (additive) $\eps$, while incurring only a $\tildeO{\log \frac{1}{\eps}}$ factor blow up in the number of queries.
Furthermore, \cite{PRR:06} shows that both tolerant testing and distance approximation are no harder than agnostic learning.
We prove the following general result.

\begin{restatable}{theorem}{toltesthighdim}\label{theo:tol:test:d:na:full}
There exist
\begin{itemize}
  \item a non-adaptive (fully) tolerant tester for $k$-monotonicity of functions $f\colon [n]^d\to\{0,1\}$ 
  with query complexity $q(n,d,\eps_1,\eps_2,k)=\tildeO{ \frac{1}{(\eps_2-\eps_1)^2} \left(\frac{5kd}{\eps_2-\eps_1}\right)^d }$;
  \item a non-adaptive tolerant tester for $k$-monotonicity of functions $f\colon [n]^d\to\{0,1\}$ with query complexity $q(n,d,\eps_1,\eps_2,k)=2^{\tildeO{k\sqrt{d}/(\eps_2 - 3\eps_1)^2}}$, under the restriction that $\eps_2 > 3\eps_1$.
\end{itemize}
\end{restatable}

To the best of our knowledge, the only previous results for tolerant testing  for monotonicity on $[n]^d$ are due to 
 Fattal and Ron~\cite{FattalR:10}. They give both additive and multiplicative distance approximations algorithms, and obtain $O(d)$-multiplicative  and  $\eps$-additive approximations with query complexity $\poly(\frac{1}{\eps})$. While very efficient, there results only give fully tolerant testers for dimensions $d=1$ and $d=2$. Our results generalize the work of ~\cite{FattalR:10} showing  existence of tolerant testers for $k$-monotonicity (and hence for monotonicity) for any dimension $d\geq 1$, and any $k\geq 1$, but  paying the price in the query complexity.
 
  As a consequence to~\autoref{theo:tol:test:d:na:full}, we make progress on an open problem of Berman \etal.~\cite{BermanRY:14}, as explained next.

\paragraph{Testing $k$-monotonicity under $\lp[p]$ distance} The property of being a monotone Boolean function has a natural extension to real-valued functions. Indeed, a real-valued function defined over a finite domain $D$ is monotone  if $f(x)\leq f(y)$ whenever $x\preceq y$. For real-valued functions the more natural notion of distance is $\lp[p]$ distance, rather than Hamming distance. The study of monotonicity has been extended to real-valued functions in a recent work by Berman \etal.~\cite{BermanRY:14}. They give tolerant testers for grids of dimension $d=1$ and $d=2$, and leave open the problem of extending the results to general $d$, as asked explicitly at the recent Sublinear Algorithms Workshop 2016~\cite{Sublinear:open:70}.

 We make progress towards solving this open problem, by combining our~\autoref{theo:tol:test:d:na:full}  with a reduction from $\lp[p]$ testing to Hamming testing from \cite{BermanRY:14}.

\begin{restatable}{theorem}{toltestlone}\label{coro:tol:test:l1:na}
There exists a non-adaptive tolerant $\lp[1]$-tester for monotonicity of functions $f\colon [n]^d\to\{0,1\}$ with query complexity
\begin{itemize}
  \item $\tildeO{ \frac{1}{(\eps_2-\eps_1)^2} \left(\frac{5d}{\eps_2-\eps_1}\right)^d }$, for any $0\leq \eps_1 < \eps_2 \leq 1$;
  \item $2^{\tildeO{\sqrt{d}/(\eps_2-3\eps_1)^2}}$, for any $0\leq 3\eps_1 < \eps_2 \leq 1$.
\end{itemize} 
\end{restatable}
 
\makeatletter{}\subsection{Proofs overview and technical contribution}\label{sec:intro:overview}

\paragraph{Structural properties and the separation between testing and learning on $\{0,1\}^d$. }

We first observe that basic structural properties, such as {\em extendability} (i.e. the feature that a function that is monotone on a sub-poset of $[n]^d$ can be extended into a monotone function on the entire poset domain), and properties of the {\em violation graph} (i.e., the graph whose edges encode the violations to monotonicity), extend easily to $k$-monotonicity (see \autoref{app:poset:structural}). These properties help us to argue the separation between testing and learning (\autoref{theo:test:d:os:na:ub}). However, unlike the case of monotonicity testing, these properties do not seem to be enough for showing  upper bounds that grow  polynomially in $d$.

\paragraph{Grid coarsening and  testing by implicit/explicit learning.}
One pervading technique, which underlies all the hypergrid  upper bounds in this work, is that of \emph{gridding}: i.e., partitioning the domain into ``blocks'' whose size no longer depends on the main parameter of the problem, $n$. This technique generalizes the approach of \cite{FattalR:10} who performed a similar gridding for dimension $d=2$.
By simulating query access to the ``coarsened'' version of the unknown function (with regard to these blocks), we are able to leverage methods such as testing-by-learning (either fully or partially learning the function), or reduce our testing problem to a (related) question on these nicer ``coarsenings.'' (The main challenge here lies in providing efficient and consistent oracle access to the said coarsenings.)\marginnote{Terminology: \emph{gridding}, \emph{coarsening}}

At a high-level, the key aspect of $k$-monotonicity which makes this general approach possible is reminiscent of the concept of \emph{heredity} in property testing. Specifically, we rely upon the fact that ``gridding preserves $k$-monotonicity:'' if $f$ \emph{is} $k$-monotone, then so will be its coarsening $g$ -- but now $g$ is much simpler to handle. This  allows us to trade the domain $[n]^d$ for what is effectively $[m]^d$, with $m \ll n$. We point out that this differs from the usual paradigm of \emph{dimension reduction}: indeed, the latter would reduce the study of a property of functions on $[n]^d$ to that of functions on $[n]^{d^\prime}$ for $d^\prime \ll d$ (usually even $d^\prime = 1$) by projecting $f$ on a lower-dimensional domain. In contrast, we do not take the dimension down, but instead reduce the size of the \emph{alphabet}. \new{Moreover, it is worth noting that this gridding technique is also orthogonal to that of \emph{range reduction}, as used e.g. in~\cite{DGLRRS99}. Indeed, the latter is a reduction of the range of the function from $[R]$ to $\{0,1\}$, while gridding is only concerned about the domain size.}

\paragraph{Estimating the support of distributions.} Our proof of the $\poly(1/\eps)$ upper bound for testing $k$-monotonicity on the line (\autoref{theo:test:d1:2s:na:ub}) rests upon an unexpected connection to \emph{distribution testing}, namely to the question of support size estimation of a probability distribution. In more detail, we describe how to reduce $k$-monotonicity testing to the support size estimation problem in (a slight modification of) the \emph{Dual access model} introduced by Canonne and Rubinfeld~\cite{CanonneR:14}, where the tester is granted samples from an unknown distribution as well as query access to its probability mass function.

For our reduction to go through, we first describe how any function $f\colon[n]\to\{0,1\}$ determines a probability distribution $D_f$ (on $[n]$), whose effective support size is directly related to the $k$-monotonicity of $f$. We then show how to implement dual access to this $D_f$ from queries to $f$: in order to avoid any dependence on $k$ and $n$ in this step, we resort both to the gridding approach outlined above (allowing us to remove $n$ from the picture) and to a careful argument to ``cap'' the values of $D_f$ returned by our simulated oracle. Indeed, obtaining the exact value of $D_f(x)$ for arbitrary $x$ may require $\Omega(k)$ queries to $f$, which we cannot afford; instead, we argue that only returning $D_f(x)$ whenever this value is ``small enough'' is sufficient. Finally, we show that implementing this ``capped'' dual access oracle is possible with no dependence on $k$ whatsoever, and we can now invoke the support size estimation algorithm of~\cite{CanonneR:14} to conclude.

\paragraph{Fourier analysis on the hypergrid.}  We give an algorithm for fully tolerantly testing $k$-monotonicity whose query complexity in exponential in $d$.  We also describe an alternate tester (with a slightly worse tolerance guarantee) whose query complexity is instead exponential in $\widetilde{O}(k\sqrt{d})$ for constant distance parameters.  As mentioned above, we use our gridding approach combined with tools from learning theory. Specifically, we employ an agnostic learning algorithm of~\cite{KalaiKMS:08} using polynomial regression.  Our coarsening methods allow us to treat the domain as if it were $[m]^d$ for some $m$ that is independent of $n$.  To prove that this agnostic learning algorithm will succeed, we turn to Fourier analysis over $[m]^d$.  We extend the bound on average sensitivity of $k$-monotone functions over the Boolean hypercube from~\cite{BlaisCOST:15} to the hypergrid, and we show that this result implies that the Fourier coefficients are concentrated on ``simple'' functions. 
 
\makeatletter{}\subsection{Discussion and open problems}\label{sec:intro:discussion}

This is the first work to study $k$-monotonicity,  a natural and well-motivated generalization of monotonicity. Hence this work opens up many intriguing questions in the area of property testing, with potential applications to learning theory, circuit complexity and cryptography.

\noindent As previously mentioned, the main open problem prompted by our work is the following:

\begin{quote}
{\em Can $k$-monotonicity on the hypercube $\{0,1\}^d$  be tested with $\poly(d^k)$ queries?}
\end{quote}

A natural $1$-sided tester for $k$-monotonicity is a {\em  chain tester}: it queries points along a random chain, and rejects only if it finds a violation to $k$-monotonicity, in the form of a sequence  $x_1\preceq x_2 \preceq \ldots \preceq x_{k+1}$ in $\{0,1\}^d$, such that $f(x_1)=1$ and $f(x_i)\ne f(x_{i+1})$. In particular, the testers in~\cite{GGLRS00,CS:13b,ChenST:14,KhotMS:15} all directly imply a chain tester. We conjecture that 
 there exists a chain tester for $k$-monotonicity that succeeds with probability $d^{-O(k)}$.
 
Another important open question concerns the hypergrid domain, and in particular it pushes for a significant strengthening of~\autoref{coro:tol:test:d:na} and~\autoref{coro:tol:test:l1:na}:

\begin{quote}
{\em Can $k$-monotonicity on the hypergrid $[n]^d$  be (tolerantly) tested with $2^{o_k(\sqrt{d})}$ queries?}
\end{quote}
Answering this question would imply further progress on the $L_1$-testing question for monotonicity, left open in~\cite{BermanRY:14, Sublinear:open:70}.\medskip

There also remains the question of establishing two-sided lower bounds that would go beyond those of monotonicity. Specifically:
\begin{quote}
{\em Is there an $d^{\Omega(k)}$-query \emph{two-sided} lower bound for $k$-monotonicity on the hypercube $\{0,1\}^d$?}
\end{quote}

In this work we also show surprising connections to distribution testing (e.g. in the proof of~\autoref{theo:test:d1:2s:na:ub}), and to testing union of intervals and testing surface area (as discussed in~\autoref{ssec:line:twosided}). An intriguing direction is to generalize this connection to union of intervals and surface area in higher dimensions, to leverage or gain insight on $k$-monotonicity on the $d$-dimensional hypergrid.\medskip

Finally, while we only stated here a few directions, we emphasize that every question that is relevant to monotonicity is also relevant and interesting in the case of $k$-monotonicity.

\subsection{Related work}
\makeatletter{}
 As mentioned, $k$-monotonicity has deep connections with the notion of \emph{negation complexity} of functions, which is the minimum number of negation gates needed in
a circuit to compute a given function. The power of negation gates is intriguing and far from being understood in the context of circuit lower bounds. 
Quoting from Jukna's book~\cite{Jukna:12}, 
{\em the main difficulty in proving nontrivial lower bounds on the size of circuits using \textsf{AND}, \textsf{OR}, and \textsf{NOT} is the presence of \textsf{NOT} gates: we already know how to prove even exponential lower bounds for monotone functions if no \textsf{NOT} gates are allowed. The effect of such gates on circuit size remains to a large extent a mystery.}

This gap  has motivated the study of circuits with \emph{few} negations. Two notable works successfully extend lower bounds in the monotone setting to negation-limited setting:  in~\cite{AmanoMaruoka:clique}, Amano and Maruoka show superpolynomial circuit lower bounds for $(1/6)\log{\log{n}}$ negations using the $\textsc{Clique}$ function; and recently the breakthrough work of  Rossman~\cite{Rossman:15} establishes circuit lower bounds for $\textsf{NC}^1$ with roughly $\frac{1}{2}\log{n}$ negations by drawing upon his lower bound for monotone $\textsf{NC}^1$.   
 
The divide between the understanding of monotone and non-monotone computation exists in general: while we usually have a fairly good understanding of the monotone case, many things get murky or fail to hold even when a single negation gate is allowed.  In order to get a better grasp on negation-limited circuits, a body of recent work has been considering this model in various contexts: Blais~\etal.~\cite{BlaisCOST:15} study negation-limited circuits from a computational learning viewpoint,  Guo~\etal.~\cite{GuoMOR:15} study the possibility of implementing cryptographic primitives using few negations, and Lin and Zhang~\cite{LinZ:16} are interested in verifying whether some classic Boolean function conjectures hold for the subset of functions computed by negation-limited circuits.

Many of these results implicitly or explicitly rely on a simple but powerful tool: the decomposition of negation-limited circuits into a composition of some ``nice'' function with monotone components. Doing so enables one to apply results on separate monotone components, and finally to carefully combine the outcomes (e.g.,~\cite{GuoK:15}).
 Though these techniques can yield results for as many as $O(\log{n})$ negations, they also leave open surprisingly basic questions:
\begin{itemize}
\item~\cite{BlaisCOST:15}  Can we have an efficient weak learning algorithm for functions computed by circuits with a \emph{single} negation?  
\item ~\cite{GuoMOR:15} Can we obtain pseudorandom generators when allowing only a \emph{single} negation?\end{itemize}

In contexts  where the circuit size is not the quantity of interest, the equivalent notion of $2$-monotone functions is more natural than that of circuits allowing only one negation. Albeit seemingly simple, even the class of $2$-monotone functions remains largely a mystery: as exemplified above, many basic yet non-trivial questions, ranging from the structure of their Fourier spectrum to their expressive power of $k$-monotone functions, remain open. 
 
\makeatletter{}\subsection{Organization of the paper}

After recalling some notations and definitions in~\autoref{sec:prelim}, we consider the case of the Boolean hypercube in~\autoref{sec:cube}, where we establish lower bounds on testing $k$-monotonicity of functions $f\colon\{0,1\}^d \to \{0,1\}$ for both one- and two-sided algorithms, \new{and provide an algorithm which ``beats'' the testing-by-learning approach, showing that testing is provably easier than learning.}

Next, we establish our results for functions on the line in~\autoref{sec:line}, starting with the lower and upper bounds for one-sided testers before turning in~\autoref{ssec:line:twosided} to the two-sided upper bound of~\autoref{theo:test:d1:os:a:lb}. We then describe in~\autoref{sec:grid} our results for functions on the grid $[n]^2$, focusing on the case $k=2$; and discussing possible extensions in~\autoref{ssec:grid:kk}. 

\autoref{sec:high:dim} contains our general algorithms for $k$-monotonicity on the hypergrid $[n]^d$, for arbitrary $k$ and $d$. We prove~\autoref{theo:tol:test:d:na:full} in two parts. We establishing its first item (general tolerant testing algorithm with exponential dependence in $d$) in~\autoref{ssec:high:dim:full:tolerant} (\autoref{prop:tester:exp:in:d}). The second item (with query complexity exponential in $k\sqrt{d}$) is proven in~\autoref{ssec:high:dim:agnostic}, where we analyze the Fourier-based tolerant tester of~\autoref{prop:tester:exp:in:rootd}. We then apply these results to the question of tolerant $\lp[1]$-testing of monotonicity in~\autoref{sec:application:l1}, after describing a reduction between monotonicity of functions $[n]^d\to [0,1]$ and of $[n]^{d+1}\to \{0,1\}$.\medskip

Except maybe~\autoref{sec:application:l1} which depends on~\autoref{sec:high:dim}, all sections are independent and self-contained, and the reader may choose to read them in any order.

\section{Preliminaries}\label{sec:prelim}
\makeatletter{}\renewcommand{\domain}{\mathcal{X}}
\newcommand{\range}{\mathcal{Y}}

We denote by $\log$ the binary logarithm, and use $\tildeO{\cdot}$ to hide polylogarithmic factors in the argument (so that $\tildeO{f} = \bigO{f \log^c f}$ for some $c \geq 0$).

Given two functions $f,g\colon\domain\to\range$ on a finite domain $\domain$, we write $\dist{f}{g}$ for the (normalized) Hamming distance between them, i.e.
\[
    \dist{f}{g} = \frac{1}{\abs{\domain}}\sum_{x\in\domain} \indic{f(x)\neq g(x)} = \probaDistrOf{x\sim{\domain}}{ f(x) \neq g(x) }
\]
where $x\sim{\domain}$ refers to $x$ being drawn from the uniform distribution on $\domain$. A \emph{property} of functions from $\domain$ to $\range$ is a subset $\property\subseteq \domain^{\range}$ of these functions; we define the distance of a function $f$ to $\property$ as the minimum distance of $f$ to any $g\in\property$:
\[
  \dist{f}{\property} = \inf_{g\in\property} \dist{f}{g}.
\]

For some of our applications, we will also use another notion of distance specific to real-valued functions, the $L_1$ distance (as introduced in the context of property testing in~\cite{BermanRY:14}). For $f,g\colon\domain\to[0,1]$, we write
\[
    L_1(f,g) = \frac{1}{\abs{\domain}} \sum_{x\in\domain} \abs{ f(x) - g(x) } = \shortexpect_{x\sim{\domain}}[ \abs{ f(x) - g(x) } ] \in [0,1]
\]
and extend the definition to $L_1(f, \property)$, for $\property\subseteq\domain^{[0,1]}$, as before.

\paragraph{Property testing.} We recall the standard definition of testing algorithms, as well as some terminology:
\begin{definition}\label{def:testing}
Let $\property$ be a property of functions from $\domain$ to $\range$.  A \emph{$q$-query testing algorithm for $\property$} is a randomized algorithm $\Tester$ which takes as input $\eps\in(0,1]$ as well as query access to a function $f\colon\domain\to\range$.  After making at most $q(\eps)$ queries to the function, \Tester either outputs \accept or \reject, such that the following holds:
\begin{itemize}
\item if $f \in \property$, then \Tester outputs \accept with probability at least $2/3$; \hfill (Completeness)
\item if $\dist{f}{\property} \geq \eps$, then \Tester outputs \reject  with probability at least $2/3$; \hfill (Soundness)
\end{itemize}
where the probability is taken over the algorithm's randomness. If the algorithm only errs in the second case but accepts any function $f\in\property$ with probability $1$, it is said to be a \emph{one-sided} tester; otherwise, it is said to be \emph{two-sided}. Moreover, if the queries made to the function can only depend on the internal randomness of the algorithm, but not on the values obtained during previous queries, it is said to be \emph{non-adaptive}; otherwise, it is \emph{adaptive}.
\end{definition}

Additionally, we will also be interested in \emph{tolerant} testers -- roughly, algorithms robust to a relaxation of the first item above:
\begin{definition}\label{def:tol:testing}
  Let $\property$, $\domain$, and $\range$ be as above. A \emph{$q$-query tolerant testing algorithm for $\property$} is a randomized algorithm $\Tester$ which takes as input $0 \leq \eps_1 < \eps_2 \leq 1$, as well as query access to a function $f\colon\domain\to\range$. After making at most $q(\eps_1,\eps_2)$ calls to the oracle, \Tester outputs either \accept or \reject, such that the following holds:
  \begin{itemize}
    \item if $\dist{f}{\property} \leq \eps_1$, then \Tester outputs \accept with probability at least $2/3$; \hfill (Completeness)
    \item if $\dist{f}{\property} \geq \eps_2$, then \Tester outputs \reject  with probability at least $2/3$; \hfill (Soundness)
  \end{itemize}
where the probability is taken over the algorithm's randomness. The notions of one-sidedness and adaptivity of~\autoref{def:testing} extend to tolerant testing algorithms as well.
\end{definition}

Note that as stated, in both cases the algorithm ``knows'' $\domain,\range$, and $\property$; so that the query complexity $q$ can be parameterized by these quantities. More specifically, when considering $\domain=[n]^d$ and the property $\property$ of $k$-monotonicity, we will allow $q$ to depend on $n,d$, and $k$. Finally, we shall sometimes require a probability of success $1-\delta$ instead of the (arbitrary) constant $2/3$; by standard techniques, this can be obtained at the cost of a multiplicative $\bigO{\log({1}/{\delta})}$ in the query complexity.

\paragraph{PAC and agnostic learning \cite{Valiant84}} A learning  algorithm ${\cal A} $ for a {\em concept class}  $\class$ of functions $f\colon\domain\to\range$ (under the uniform distribution) is given parameters $\eps, \delta>0$ and sample access to some target function $f\in\class$ \emph{via}  labeled samples $\langle x, f(x)\rangle$, where $x$ is drawn uniformly at random from $\domain$. The algorithm should output a hypothesis $h\colon\domain\to\range$ such that $\dist{h}{f}\leq \eps$ with probability at least $1-\delta$. The algorithm is \emph{efficient} if it runs in time $\poly(n, 1/\eps, 1/\delta)$. If ${\cal A}$ must output  $h\in \class$ we say it is a \emph{proper learning algorithm}, otherwise,  we say it is an \emph{improper learning} one. 

Moreover, if $\Algo$ still succeeds when $f$ does not actually belong to $\class$, we say it is an \emph{agnostic learning algorithm}. Specifically, the hypothesis function $h$ that it outputs must satisfy $\dist{f}{g} \leq \opt_f + \eps$ with probability at least $1-\delta$, where $\opt_f = \min_{g\in\class} \dist{f}{g}$.

\section{On the Boolean hypercube}\label{sec:cube}
\makeatletter{}
\newcommand{\poset}{\ensuremath{\mathcal{P}}\xspace}
\newcommand{\rank}[1]{\operatorname{rank}_{#1}}
\newcommand{\ds}{\displaystyle}
In this section, we focus on $k$-monotonicity of Boolean functions over the hypercube $\{0,1\}^d$. We begin in~\autoref{ssec:cube:upper} with a tester with query complexity $2^{\tildeO{\sqrt{d}}}$, establishing a strict separation between learning and testing. \autoref{ssec:cube:lower} is then dedicated to our lower bounds on $k$-monotonicity testing.

\subsection{Upper bound: beating the learning approach}\label{ssec:cube:upper}

\new{In this section, we prove the following theorem:}\footnote{We note that this result is only interesting in the regime $k\leq \sqrt{d}$: indeed, for $k=\bigOmega{\sqrt{d}\log\frac{1}{\eps}}$ \emph{every} function is $\eps$-close to $k$-monotone.}

\testcubeonesidednaub*

\noindent Let us recall the following standard fact.

\begin{fact} \label{fact:all-points-in-middle}
There exists an absolute constant $C>0$ such that the number of points of  $\{0,1\}^d$ that do not have integer weights in  the  \emph{middle  levels}  $[\frac{d}{2} - \sqrt d \log \frac{C}{\eps}, \frac{d}{2} + \sqrt d \log \frac{C}{\eps}]$  is at most $\eps 2^{d-1}$.
\end{fact}

\noindent With that fact in our hands, we now describe the following tester. 

\begin{enumerate}[(1)]
  \item Sample $O(1/\eps)$ random points from the middle levels 
  \item For each of the queries in the first step, query all points with Hamming weight in the middle levels which fall in the subcube below  and in the subcube above each such random point. We call each of these $O(1/\eps)$ collections of queries a \emph{superquery }.
\end{enumerate}

\noindent The key idea behind the tester is the extendability lemma, ~\autoref{lemma:k-mono-extendable}. The tester tries to find one of the ``violated hyperedges'' from the matching of violations that we know exists from ~\autoref{lemma:k-mono-extendable}.

\noindent Now, we analyze the tester.

\begin{proofof} {\autoref{theo:test:d:os:na:ub}}

Suppose $k$ is odd.\footnote{The case $k$ even is similar. In this case, one may assume that $f$ evaluates to $0$ on points with hamming       weight everywhere outside the middle levels.}
In this case, given a function $f$ $\eps$-far from $k$-monotonicity we can assume without loss of generality that $f$ and equals $0$ on points with Hamming weight  $<\frac{d}{2} - \sqrt d \log \frac{C}{\eps}$, and equals $1$ on points with Hamming weight $>\frac{d}{2}  + \sqrt d \log  \frac{C}{\eps}$. The resulting function is still $\frac{\eps}{2}$-far from being $k$-monotone.\\

		Now, by ~\autoref{lemma:k-mono-extendable}, we know that there exists a matching $M_f$ of violations to $k$-monotonicity in $f$ of size $\frac{1}{2} \cdot \frac{\eps 2^d}{k}$. The set of vertices that participate in these violations has cardinality $$|M_f| \cdot \frac{\eps 2^d}{k} = \frac{\eps (k+1) 2^d}{k} \geq \eps 2^{d-1}$$ With constant probability, in $O(1/\eps)$ queries in the first step, the tester queries a vertex that belongs to some violation from $M_f$. Then in the next step, the tester finds this violation. Now we bound the number of queries made in a single superquery. Since it involves only $2\sqrt d \log \frac{C}{\eps}$ levels of the cube, the number of points  queried in a single superquery is no more than $b = d^{O(\sqrt d \log \frac{1}{\eps})}$. The total query complexity of the tester can therefore be upperbounded by $b/\eps = d^{O(\sqrt d \log \frac{1}{\eps})}$. We emphasize that the number of queries made by this tester has no dependence on $k$.

\end{proofof}

\subsection{Lower bounds}\label{ssec:cube:lower}

\new{We now turn to lower bounds for testing $k$-monotonicity of Boolean functions over the hypercube $\{0,1\}^d$.}  In~\autoref{sec:1s-na-lb}, we show that for constant $k$ any one-sided non-adaptive tester for $k$-monotonicity requires $\Omega(d^{k/4})$ queries , generalizing the $\bigOmega{\sqrt{d}}$ lower bound for monotonicity due to Fischer\ \etal.~\cite{FLNRRS:02}. This bound suggests the problem become strictly harder when $k$ increases: specifically, for $k>2$ testing $k$-monotonicity requires $\omega(\sqrt{d})$ queries, while an $O(\sqrt{d})$ one-sided non-adaptive upper bound holds for monotonicity testing~\cite{KhotMS:15}. 

We then describe in~\autoref{sec:2s-lb} a general reduction from monotonicity testing to $k$-monotonicity testing, for arbitrary constant $k$.
This blackbox reduction allows us to carry any lower bound (possibly 2-sided or adaptive) for monotonicity testing to $k$-monotonicity.
In particular, combining it with the recent lower bounds~\cite{ChenDST:15,BelovsB:15} for 2-sided monotonicity testing, 
we obtain an $\Omega(d^{1/2-o(1)})$ lower bound for non-adaptive $k$-monotonicity testers, and an $\bigOmega{d^{1/4}}$ lower bound for adaptive ones.

\subsubsection{One-sided lower bounds}\label{sec:1s-na-lb}
\onesidedhypercube*
Consider the family of functions $\{f_S\colon S\subseteq[d] \text{ of size $t$}\}$ where $t\geq k$ is a parameter to be determined later and $f_S$ is a truncated anti-parity over $t$ input coordinates indexed by $S$, namely,

 \begin{align*}
  f_{S}=
  \begin{cases}
    \overline{\oplus_{i\in S} x_{i}}   & \text{if } \abs{\abs{x}-d/2}\leq \sqrt{d},   \\
    0    & \text{otherwise}.
  \end{cases}
\end{align*}

\autoref{thm:1-sided-non-adaptive-lb} immediately follows from the two claims below. In particular,by \autoref{cl:k_dis}, for $t\leq\sqrt{d}$ and $t=4k^2$, $f_S$ is $\Omega(1)$-far from any $k$-monotone function, and  by \autoref{cl:k_bound}, to reject every $f_S$ with $\bigOmega{1}$ probability, $q$ needs to be at least $d^{k/4}\cdot (\frac{k}{2e^2t})^{k/2} = \Omega(d/k^2)^{k/4}$ when $t=4k^2$. 

\begin{claim}
	\label{cl:k_bound}
	For any  non-adaptive $q$-query algorithm $\Algo$, there exists  $f_S$ where $\abs{S}=t$ such that $\Algo$ reveals a violation on $f_{S}$
	with probability at most $q^2 \binom{2\sqrt{d}}{k}\binom{t}{k}/\binom{d}{k}$.
\end{claim}

\begin{claim}
	\label{cl:k_dis}
	For $k\leq t\leq \sqrt{d}$, $f_{S}$ is $\Omega(\sum_{i=0}^{\left\lfloor{\frac{t-k-1}{2}}\right \rfloor}\binom{t}{i}/2^t)$-far from any $k$-monotone functions.
\end{claim}

\begin{proof} [Proof of~\autoref{cl:k_bound}]
Let $Q$ be an arbitrary set of $q$ queries. Without loss of generality, we assume  every query $z$ in $Q$ has Hamming weight $\abs{z} \in [\frac{d}{2}-\sqrt{d}, \frac{d}{2}+\sqrt{d}]$. We define $Q_{x,z} = \setOfSuchThat{ y\in Q }{x\preceq y\preceq z }$ and $\mathrm{Rej}(Q)=\setOfSuchThat{f_{S} }{ Q \text{ contains a violation for $f_S$}}$. Note that $\mathrm{Rej}(Q)=\bigcup_{(x, z)\colon x\preceq z\in Q} \mathrm{Rej}(Q_{x,z})$.
Hence we can bound the size of $\mathrm{Rej}(Q)$ by
\begin{equation}\label{eq:lb:os:na:1}
  \abs{\mathrm{Rej}(Q)} \leq \sum_{(x,z)\in Q^2\colon x\preceq z} \abs{\mathrm{Rej}(Q_{x,z})}.
\end{equation}
Fix any $x\preceq z\in Q$. Because any violation for $f_{S}$ in $Q_{x,z}$ contains at least two points $x'$ and $z'$ such that  $x\preceq x'\preceq z'\preceq z$, and
$x'_{S'} = 0^k$ and $z'_{S'} = 1^k$, we have $x_{S'} = 0^k$ and $z_{S'} = 1^k$ for some $S'\subseteq S$ of size $k$. Note that
$x$ and $z$ differ in at most $2\sqrt{d}$ coordinates so that there are at most $\binom{2\sqrt{d}}{k}$ distinct $S'$ on which $x_{S'}$ is $0^k$
and $z_{S'}$ is $1^k$. Moreover, for each $S'$ there are at most $\binom{d-k}{t-k}$ distinct $S$ such that $S'\subseteq S$. Therefore we can bound the size of $\mathrm{Rej}(Q_{x,z})$ by
\begin{equation}\label{eq:lb:os:na:2}
  \abs{\mathrm{Rej}(Q_{x,z})} \leq \binom{2\sqrt{d}}{k} \binom{d-k}{t-k}.
\end{equation}
Combining \eqref{eq:lb:os:na:1} and \eqref{eq:lb:os:na:2} , $\abs{\mathrm{Rej}(Q)} \leq q^2 \binom{2\sqrt{d}}{k}\binom{d-k}{t-k}$.  It follows that for any non-adaptive algorithm making at most $q$ queries, 
\[\sum_{S\subseteq[d]:\abs{S}=t} \Pr[ \Algo\text{ reveals a violation for } f_{S}]\leq \shortexpect[ \abs{\mathrm{Rej}(Q)}]\leq q^2 \binom{2\sqrt{d}}{k}\binom{d-k}{t-k}.\]
Hence there exists $f_{S}$ such that $\Pr[ \Algo \text{ reveals a violation for } f_{S}]\leq q^2  \frac{\binom{2\sqrt{d}}{k}\binom{d-k}{t-k}}{\binom{d}{t}}=  q^2 \frac{\binom{2\sqrt{d}}{k}\binom{t}{k}}{\binom{d}{k}}$.
\end{proof}

\begin{proof} [Proof of~\autoref{cl:k_dis}]
Let $f'_{S}$ be the closest $k$-monotone function to $f_{S}$. Let $Z$ denote the set $\{z\in \B^{d-t}: d/2-\sqrt{d}\leq \abs{z}\leq d/2 + \sqrt{d}-t\}$. For any $t\leq \sqrt{d}$, $[\frac{d-t}{2} - \frac{\sqrt{d-t}}{2}, \frac{d-t}{2} +\frac{\sqrt{d-t}}{2}]$ is contained in $[d/2-\sqrt{d}, d/2 + \sqrt{d}-t]$ so that $\abs{Z}=\Omega(2^{d-t})$. For any assignment $z\in Z$ on coordinates indexed by $[d]\setminus S$,  $f_{S}(\cdot , z)$ agrees with $\overline{\oplus_{i\in S} x_{i}}$ and $f'_{S}(\cdot , z)$ is $k$-monotone.  To finish the proof, it suffices to show that an anti-parity over $t$ inputs is $\sum_{i=0}^{\flr{ \frac{t-k-1}{2} }}\binom{t}{k}/2^t$-far from any $k$-monotone function over $t$ inputs. Indeed, this will imply that, for every $z\in Z$, $f_{S}(\cdot , z)$ differs from $f'_{S}(\cdot , z)$ on $\sum_{i=0}^{\flr{ \frac{t-k-1}{2} }}\binom{t}{k}$ points, and thus finally that $f_{S}$ differs from $f_{S'}$ on 	$\abs{Z}\cdot\sum_{i=0}^{\flr{ \frac{t-k-1}{2} }}\binom{t}{k}=\Omega(2^{d-t}\cdot\sum_{i=0}^{\left\lfloor{\frac{t-k-1}{2}}\right \rfloor}\binom{t}{k})$ points.
 
Now we show that an anti-parity function over $t$ inputs  $h(x_1,\dots,x_t)=\overline{\oplus_{i=1}^t x_{i}}$ is $\sum_{i=0}^{\flr{ \frac{t-k-1}{2} }}\binom{t}{i}/2^t$-far from any $k$-monotone function $g$ (over $x_1,\dots,x_t$).  We begin by noting that we can sample a random point from $\{0,1\}^t$ by first sampling a random chain $\mathcal{C}=(x^0=0^t,x^1,\dots,x^{t}=1^t)$ (from all possible chains from $0^t$ to $1^t$) then outputting $x^i$ with probability $\binom{t}{i}/2^t$. Thus the distance  between $h$ and $g$, namely $\Pr_{x}[h(x)\neq g(x)]$, is equal to
\begin{align}\label{eq:lb:os:na:3}
\Pr_{\mathcal{C},i}[h(x^i)\neq g(x^i)] 
=\shortexpect_{\mathcal{C}}\left[\sum_{i=0}^t \frac{\binom{t}{i}}{2^t} \cdot \indic{h(x^i)\neq g(x^i)}\right]
= \shortexpect_{\mathcal{C}}\left[\frac{1}{2^t}\sum_{i=1}^t \binom{t-1}{i-1} \cdot (\indic{h(x^i)\neq g(x^i)}+ \indic{h(x^{i-1})\neq g(x^{i-1})})\right]
\end{align}
where the last inequality relies on identity $\binom{t}{i}=\binom{t-1}{i-1}+\binom{t-1}{i}$. For any fixed chain $\mathcal{C}$, because $g$ alternates at most $k$ times, there are at least $t-k$ choices of $i\in[t]$ such that $g(x^{i-1})=g(x^i)$, which implies $\indic{h(x^i)\neq g(x^i)}+ \indic{h(x^{i-1})\neq g(x^{i-1})}\geq 1$ (due to $h(x^{i-1})\neq h(x^{i})$). Thus 
$\sum_{i=1}^t \binom{t-1}{i-1}\cdot (\indic{h(x^i)\neq g(x^i)}+ \indic{h(x^{i-1})\neq g(x^{i-1})})$ is at least the sum of smallest $t-k$ binomials among $\binom{t-1}{0},\dots,\binom{t-1}{t-1}$ which is \[\sum_{i=0}^{\flr{ \frac{t-k-1}{2} }}\binom{t-1}{i} + \sum_{i=0}^{\left\lfloor{\frac{t-k-2}{2}}\right \rfloor }\binom{t-1}{t-1-i} =\sum_{i=0}^{\flr{ \frac{t-k-1}{2} }}\binom{t-1}{i} + \sum_{i=0}^{\left\lfloor{\frac{t-k-2}{2}}\right \rfloor }\binom{t-1}{i} \geq \sum_{i=0}^{\flr{ \frac{t-k-1}{2} }}\binom{t}{i},\]
which implies $\Pr_{x}[h(x)\neq g(x)]\geq \frac{1}{2^t} \sum_{i=0}^{\flr{ \frac{t-k-1}{2} }}\binom{t}{i}$ by combining the above with \eqref{eq:lb:os:na:3}.

\end{proof}

 \subsubsection{Two-sided lower bounds}\label{sec:2s-lb}
The following theorem gives a construction that enables us to convert monotone functions into $k$-monotone functions, and functions that are far from monotone into functions that are far from $k$-monotone. 
\begin{theorem}~\label{theorem:balanced-blocks-preserves-distance} 
There exists an efficiently computable function $h\colon\B^{d/2}\to\B$ such that for any $g\colon\B^{d/2}\to\B$, then $g||h\colon\B^{d}\to\B$ is a Boolean function (defined below) satisfying the following.
\begin{itemize}
\item if $g$ is monotone, then $g||h$ is a $k$-monotone function;
\item if $g$ is $\eps$-far from monotone, then $g||h$ is $\Omega(\eps/k)$-far from being a $k$-monotone function;
\end{itemize}
where $(g||h)(x,y) \eqdef g(x) \oplus h(y)$ for any $x,y \in \{0,1\}^{d/2}$. 
\end{theorem}

The above theorem reduces test monotonicity to testing $k$-monotonicity (for arbitrary constant $k$) with the same number of queries to the input function.
This theorem allows us to carry any lower bound on monotonicity testing to $k$-monotonicity, while preserving the characteristics (two-sidedness, adaptivity) of the original lower bound.  
In particular, combining it with the recent recent of~\cite{ChenDST:15,BelovsB:15}, we obtain the following corollary. 
\begin{corollary} \label{theo:k-monotone-carry}
For any $c > 0$ and $k \geq 1$, there exists $\eps = \eps(k, c) > 0$ such that any 2-sided non-adaptive algorithm for testing whether $f$ is $k$-monotone or $\eps$-far from it requires $\Omega(d^{1/2 - c})$ queries.
Any 2-sided adaptive algorithm requires $\tildeOmega{d^{1/4}}$ queries.
\end{corollary}

To prove~\autoref{theorem:balanced-blocks-preserves-distance}, we prove following three claims. \autoref{cl:keep-k} and~\autoref{cl:keep-farness} show that the existence of a $(k-1)$-monotone function $h$ for which one can find a big enough set of vertex disjoint paths in the hypercube whose labelling under $h$ satisfies some specific condition will imply~\autoref{theorem:balanced-blocks-preserves-distance}. Finally in~\autoref{cl:existence-h}, we establish the existence of such 
 $h$, and~\autoref{theorem:balanced-blocks-preserves-distance} follows.

\begin{claim}\label{cl:keep-k}
Let $h\colon\B^{d}\rightarrow\B$ be a $(k-1)$-monotone function. Then for any monotone $g\colon\B^d\rightarrow\B$, $f = g || h$ is a $k$-monotone function. 
\end{claim}

\begin{claim} \label{cl:keep-farness}
Suppose there exists a $h\colon\B^{d}\rightarrow\B$ such that the following holds. There exists at least $M$ paths of length $k-1$ such that (i) all paths are vertex disjoint and (ii) for every path $y_1\preceq \dots \preceq y_k$, $h(y_1)=0$ and $h(y_i)\neq h(y_{i+1})$ for $1\leq i\leq k-1$. Then, for any $g\colon\B^d\rightarrow\B$ which is $\eps$-far from being a monotone function, the function $f = g || h$ is $\Omega(\frac{M}{2^d}\cdot \eps)$-far from being a $k$-monotone function.
\end{claim}

\begin{claim} \label{cl:existence-h}
For any constant $k$, there exists an efficient computable $h\colon\B^{d}\rightarrow\B$ such that $h$ is a $(k-1)$-monotone function and $h$ contains at least $\frac{(1 - o_d(1)) 2^d}{k}$ paths of length $k-1$ such that all paths are vertex disjoint and for every path $y_1\preceq \dots \preceq y_k$, $h(y_1)=0$ and $h(y_i)\neq h(y_{i+1})$ for $1\leq i\leq k-1$.
\end{claim}

\begin{proof}[Proof of {\autoref{cl:keep-k}}]
Suppose $f=g || h$ is not $k$-monotone, then there exist $(x_1,y_1),\dots, (x_{k+1},y_{k+1})$ such that $(x_1,y_1)\preceq\dots\preceq (x_{k+1},y_{k+1})$, $f(x_1,y_1)=1$ and $f(x_i,y_i)\neq f(x_{i+1},y_{i+1})$ for any $1\leq i\leq k$.
Because $g$ is monotone, either $g$ is constant on $x_1,\dots, x_{k+1}$, or 
there exists  an index $1<j\leq k+1$ such that $g(x_i)=0$ for $i<j$ and $g(x_i)=1$ for $i\geq j$.
For the first case, $h(y_i)=f(x_i,y_i)$ for any $i$ and $h$ alters exactly $k$ times on points $y_1\preceq \dots \preceq y_{k+1}$. 
For the second case, $h(y_1)=f(x_1,y_1)=1$, and $h$ alters exactly $k-1$ times on $y_1\preceq \dots\preceq y_{j-1}\preceq y_{j+1}\preceq\dots\preceq y_{k+1}$. Both cases contradict $h$ being $(k-1)$-monotone.
\end{proof}

\begin{proof}[Proof of {\autoref{cl:keep-farness}}]
Let $M_h$ be the maximal set of paths of length $k$ such that all paths are vertex disjoint and for every path $y_1\preceq \dots \preceq y_k$, $h(y_1)=0$ and $h(y_i)\neq h(y_{i+1})$ for $1\leq i\leq k-1$. Let 
$M_g$ be the maximal set of pairs such that all pairs are vertex disjoint and every pair  $x_1\preceq x_2$ is a violation for $g$, i.e., $g(x_1)=1$ and $g(x_2)=0$. 
For each path $(y_1\preceq\dots\preceq y_{k})\in M_h$ and any pair $(x_1\preceq x_2)\in M_g$, it is easy to see following path is a violation for $f||g$ being $k$-monotone :                 
    \[ (x_1,y_1),\dots,(x_1,y_k), (x_2, y_k).\]
Let $f'$ be the closest $k$-monotone function to $f$. For each violating path, $f$ and $f'$ differ on at least $1$ point. Because both $M_h$ and $M_g$ are vertex disjoint, violating paths constructed by taking every path in $M_h$ and every pair in $M_g$ are vertex disjoint.  Thus $f$ and $f'$ differ on at least $\abs{M_h}\times\abs{M_g}$ points and $f$ is $(\abs{M_h}\cdot\abs{M_g}/2^{2d})$-far from $f'$.  
It is known (\cite{GGLRS00}) that for any $g\colon\B^d\to\B$ which is $\eps$-far from monotone, $\abs{M_g}\geq 2^{d-1}\eps$. The desired conclusion follows.
\end{proof}

\begin{proof} [Proof of {\autoref{cl:existence-h}}]
Let $B_1, B_2, \ldots, B_k$ be consecutive blocks each consisting of consecutive layers of the hypercube such that for each $i \in [k]$, $$ (1 - \frac{k}{\sqrt d}) \frac{2^d}{k} \leq \abs{B_i} \leq (1 + \frac{k}{\sqrt d}) \frac{2^d}{k}.$$  
Because $k$ is a constant and every layer contains at most $2^d/\sqrt{d}$ points, we can always greedily find $B_1,\dots, B_k$ one by one. Let $h$ be a function such that $h$ has constant value $(i+1 \mod 2)$ on block $B_i$ for $1\leq i\leq k$. 
It is easy to see  that $h$ is $(k-1)$-monotone. Next we  prove for any $1\leq j\leq k$, $B_1,\dots, B_j$ contain at least $\frac{(1 - o_d(1)) 2^d}{k}$ vertex disjoint paths of length $j-1$ such the $i$th point on every path is in $B_i$. \autoref{cl:existence-h} follows from the case $j=k$. 

 For $j=1$, the statement holds by taking all points in $B_1$. For $j>1$, assume that $B_1,\dots, B_{j-1}$ contain a set $P_{j-1}$ of such $\frac{(1 - o_d(1)) 2^d}{k}$ vertex disjoint paths of length $j-2$. Let $M$ be the maximal matching between $B_{j-1}$ and $B_{j}$ and let $P_{j}$ be the set of paths of length $j-1$ constructed in following way:  for each path in $P_{j-1}$ with an endpoint  $y_{j-1}$,  if there exists $y_j$ such that $(y_{j-1},y_j)\in M$, we add the path appended with $y_j$ into $P_j$.  Because no points in $B_j$ will be added into two different paths in $P_j$ and $P_{j-1}$ are vertex disjoint, paths in $P_j$ are vertex disjoint.  
 
Now we show $\abs{M} = \min(\abs{B_{j-1}},\abs{B_j})$ which implies $\abs{P_j} \geq  \frac{(1 - o_d(1)) 2^d}{k}$.  Suppose $\abs{B_{j-1}} \leq \abs{B_{j}}$ (the argument is analogous in the other case).
For any subset $S$ of $B_{j-1}$, let $f_S$  be the indicator function of the upper closure of $S$ denoted as $N(S)$.  It is not hard to check that $f_S$ is monotone and thus 
				$\Pr[f_{S}(x)=1 | x\in B_{j}] \geq \Pr[f_{S}(x)=1 | x\in B_{j-1}].$ 
It follows \[\abs{N(S)\cap B_j} = \abs{B_j} \Pr[f_{S}(x)=1 | x\in B_{j}]  \geq   \abs{B_{j-1}} \Pr[f_{S}(x)=1 | x\in B_{j-1}] \geq \abs{S}.\]
By Hall's theorem, $\abs{M}=\abs{B_{j-1}}$.   By similar argument, we can show $\abs{M}=\abs{B_{j}}$ when $\abs{B_{j-1}}>\abs{B_j}$. Thus  $\abs{M} = \min(\abs{B_{j-1}},\abs{B_j})$.
\end{proof}

\section{On the line}\label{sec:line}
\makeatletter{}In this section we prove our results on testing $k$-monotonicity on the line, that is of functions $f\colon[n]\to\{0,1\}$. We start with~\autoref{theo:test:d1:os:na:ub}, which establishes that this can be done non-adaptively with one-sided error, with only $\bigO{k/\eps}$ queries; we then turn to~\autoref{theo:test:d1:os:a:lb}, which shows that this is the best one can hope for if insisting on one-sidedness. The last result of this section is~\autoref{theo:test:d1:2s:na:ub}, where we show that -- perhaps unexpectedly -- \emph{two-sided} algorithms, even non-adaptive, can break this barrier and test $k$-mononicity with \emph{no} dependence on~$k$.

\subsection{Upper and lower bounds for one-sided testers}

We first prove the upper bound, restated below:
\testlineonesidednaub*

\begin{proof}
We assume that $\frac{\eps n}{50k}$ is an integer,\footnote{If not, we consider instead $\eps^\prime\eqdef\frac{50k}{n}\flr{\frac{\eps n}{50k}} > \eps-\frac{50k}{n} > \frac{\eps}{2}$ if $\eps > \frac{100k}{n}$; while if $\eps \leq \frac{100k}{n}$ we query the entire function, for a total of $n = \bigO{\frac{k}{\eps}}$ queries.} and partition the domain into $K\eqdef \frac{50k}{\eps}$ intervals of size $\frac{\eps n}{50k}$, the consecutive ``blocks'' $B_1,\dots, B_{K}$. We then define $g\colon [n]\to \{0,1,\unknown\}$ as the function constant on each block $B_i=\{b_i,\dots, b_{i+1}-1\}$, such that 
\begin{itemize}
  \item If $f(b_i)= f(b_{i+1}-1)$, then $g(j) = f(b_i)$ for all $j\in B_i$;
  \item otherwise, $g(j) = \unknown$ for all $j\in B_i$.
\end{itemize}
We say that a block $B_i$ such that $g|_{B_i} = \unknown$ is a \emph{changepoint block}\marginnote{Terminology: \emph{changepoint block}}\ for $g$. Clearly, given query access to $f$ one can obtain the value of $g$ on any point $j\in[n]$ with only two queries to $f$. Moreover, defining $\tilde{g}\colon [n]\to \{0,1\}$ to be the function obtained from $g$ by replacing $\ast$ by $0$, we observe the following:
\begin{itemize}
  \item If $f$ is $k$-monotone, then (i) so is $\tilde{g}$, and (ii) $f$ and $\tilde{g}$ differ in at most $k$ blocks (namely the changepoint blocks of $g$), so that $\dist{f}{g} \leq k\cdot \frac{1}{K} = \frac{\eps}{50}$;
  \item If $f$ is $\eps$-far from $k$-monotone, then either (i) $\tilde{g}$ is not $k$-monotone, or (ii)  $\dist{f}{\tilde{g}} > \eps$.
\end{itemize}

We start by learning $g$ (and thus $\tilde{g}$) exactly, using $2K = \bigO{{k}/{\eps}}$ non-adaptive queries. Setting $m\eqdef C\frac{k}{\eps}$, we also sample $m^\prime \sim \poisson{m}$ points\footnote{The fact that we sample $\poisson{m}$ instead of $m$ is for ease of the analysis; note that due to the tight concentration of Poisson random variables, with probability $1-\littleO{1}$ we will have $m^\prime \leq 2m$. If this does not happen, the tester can output \accept, incurring only a small additional error probability (and not affecting the one-sidedness).} $j_1,\dots, j_{m^\prime}$  independently and uniformly from $[n]$, where $C>0$ is a constant to be determined in the course of the analysis, and query the value of $f$ (and $g$) on all of them.
Then, we reject if either (i) $\tilde{g}$ is not $k$-monotone; or (ii) there exist at least $k+1$ distinct blocks which contain a sample $s_j$ such that $\tilde{g}(s_j) \neq f(s_j)$.

By definition, this tester is non-adaptive; and it is not difficult to see it accepts any $k$-monotone function with probability $1$, since in that case $f$ and $g$ (and a fortiori $\tilde{g}$) differ in at most $k$ blocks: indeed, these blocks can only be changepoint blocks for $g$, i.e. blocks where $f$ changes value.

It remains to argue soundness: we will show that if $f$ is $\eps$-far from $k$-monotone, the tester will reject with probability at least $2/3$. By the first check made, (i), we can assume in the following that $\tilde{g}$ is $k$-monotone -- as otherwise $f$ is rejected with probability $1$ -- and we need to show that (ii) will reject with probability at least $2/3$. For each block $B_i$ (where $i\in [K]$), let $p_i\in [0,\frac{1}{K}]$ be defined as the (normalized) number of points in $B_i$ on which $f$ and $\tilde{g}$ differ (we henceforth refer to such a point as a \emph{giveaway point}):\marginnote{Terminology: \emph{giveaway point}}
\[
    p_i\eqdef \frac{1}{n} \sum_{j\in B_i} \indic{f(j)\neq \tilde{g}(j)}.
\]
Since $f$ is $\eps$-far from the $k$-monotone function $\tilde{g}$, we have $\sum_{i=1}^K p_i \geq \eps$. Now, letting $Z_i$ be the indicator of the event that among the $m^\prime$ samples, at least one is giveaway point from $B_i$, and $Z=\sum_{i=1}^K Z_i$, we can write $Z_i = \indic{Y_i \neq 0}$, where the $(Y_i)_{i\in [K]}$ are \emph{independent} Poisson random variables with $Y_i \sim \poisson{mp_i}$. The expected number of blocks in which a giveaway point is sampled is then
\[
  \shortexpect Z = \sum_{i=1}^K \shortexpect Z_i = \sum_{i=1}^K \probaOf{Y_i \neq 0 }= \sum_{i=1}^K (1-e^{-mp_i})
\]
Since for every $i\in [K]$ it holds that $e^{-mp_i} \leq 1-\frac{m}{2}p_i$ (the inequality holding since $0 \leq mp_i \leq 1$, which is verified for $m \leq K$), we get
\[
  \shortexpect Z = \sum_{i=1}^K \shortexpect Z_i \geq \sum_{i=1}^K \frac{m}{2}p_i \geq \frac{m\eps}{2} \geq \frac{C}{2} k.
\]
Moreover, by a Chernoff bound, we get that 
\[
    \probaOf{ Z < \frac{C}{4}k } \leq e^{-\frac{Ck}{16}}\leq e^{-\frac{C}{16}} 
\]
which is less that $1/4$ for $C\geq 23$. Setting $C\eqdef 30$ satisfies both conditions that $m\leq K$ and $C\geq 23$, and results in a one-sided non-adaptive tester which rejects functions far from $k$-monotone with probability at least $1-\frac{1}{4} + \littleO{1} \geq \frac{2}{3}$.

\end{proof} 

\noindent Turning to the lower bounds against one-sided testers, we show the following:

\testlineonesidednalb* 

\begin{proof}
Since a lower bound of $\bigOmega{\frac{1}{\eps}}$ straightforwardly holds, we can restrict ourselves to $k\geq 8$, and $\eps < \frac{1}{12}$; moreover, we assume without loss of generality  that $\frac{\eps n}{k}$ is an integer, and partition the domain into $K\eqdef \frac{k}{\eps}$ intervals of size $\frac{\eps n}{k}$, the consecutive ``blocks'' $B_1,\dots, B_{K}$. For  $v\in\{0,1\}^{K/2}$, we define $g_v\colon [n]\to \{0,1\}$ as the function which has constant value $v_i$ on block $B_{2i-1}$ and has constant value $1$ on the remaining blocks.  

Consider the distribution over $\{g_v\}_v$ where  each coordinate of $v$ is independently set to $0$ with probability $p \eqdef 6\eps$, and $1$ otherwise.  We next show that $g_v$ is at least $\eps$-far from any $k$-monotone function with very high probability over the choice of $v$. By a Chernoff bound, with probability at least $1- e^{-pK/16}= 1 - e^{-3k/8}$, $g_v$ has at least $pK/4 = 3k/2$ blocks that are $0$ blocks.  Conditioned on this, it is easy to see that $g_v$ is $\eps$-far from $k$-monotone: indeed, to make it $k$-monotone one has to flip its value on at least $k$ blocks, and each block contains an $\eps/k$ fraction of the domain. 

Fix any deterministic adaptive algorithm with query complexity $q\leq k/(24\eps)$ queries, and denote by $x_1,\dots, x_q$ the sequence of queries made (when given query access to some function $g_v$).  Note that $x_1$ is fixed by the algorithm and that for $1<i\leq q$, $x_i$ is uniquely determined by previous answers  $f(x_1),\dots, f(x_{i-1})$. We can sample the distribution $\{g_v\}_v$ and answer queries from the given algorithm in the following ``lazy way'': first, by marking every even blocks with value $1$ and initializing a list of queried odd blocks with their values.  When a new query $x$ comes, if $x$ was previously queried or belongs to an even block, we return the corresponding stored value.  Otherwise,  we sample the value, which is $0$ with probability $6\eps$ and $1$ otherwise, for the odd block which $x$ belongs to; and mark this block as queried.  

Let $y_1,\dots, y_{r}$ be the following subsequence of $x_1,\dots, x_q$: $y_i$ is the $i$th query made into an odd block which is not queried in $y_1,\dots, y_{i-1}$.  Clearly $r\leq q$ and 
$y_1,\dots, y_{r}$ reveals a violation if and only if  number of $0$'s in corresponding answers is at least $k/2+1$.   
Note, for arbitrary $a\in\{0,1\}^{r}$, 
\begin{align*}
\probaOf{ f(y_1)=a_1,\dots,f(y_{r})= a_{r} } = \probaOf{ f(y_1)=a_1 }\cdot \prod_{i=2}^{r}  \probaCond{ f(y_i)=a_i }{ f(y_1)=a_1,\dots, f(y_{i-1})=a_{i-1} }.
\end{align*}
$y_i$ is determined by $f(y_1)=a_1,\dots, f(y_{i-1})=a_{i-1}$ and by our way of sampling $f(y_i)$, we know that for every $i$ it holds that $\probaCond{f(y_i)=a_i }{ f(y_1)=a_1,\dots, f(y_{i-1})=a_{i-1} } = 6\eps$ if  $a_i=0$ and 
$\probaCond{ f(y_i)=a_i }{ f(y_1)=a_1,\dots, f(y_{i-1})=a_{i-1} }=1-6\eps$ if $a_i=1$.   Thus.
\begin{equation}
\label{eq:4chernoffbound}
    \probaOf{ f(y_1)=a_1,\dots,f(y_{r})= a_{r} }  = (1-6\eps)^{ \abs{a} } (6\eps)^{r- \abs{a} }.
\end{equation}
Let $Y_i$ be the indicator that $f(y_i)=0$. We get that, writing $F(i,N,p)$ for the cumulative distribution function of a Binomial with parameters $N$ and $p$,
\begin{align*}
  \probaOf{ \sum_{i=1}^{r}  Y_i \geq \frac{k}{2} + 1 } &\leq 
  \sum_{a\in\{0,1\}^{r} \colon \abs{\bar{a}}\geq k/2} (1-6\eps)^{\abs{a}} (6\eps)^{r-\abs{a}} \\
  &=  \sum_{\ell=0}^{r - k/2} \binom{r}{\ell}(1-6\eps)^{\ell} (6\eps)^{r-\ell}
  = F\!\left(r-\frac{k}{2}, r, 1-6\eps\right) \\
  &= F\!\left(r(1-x), r, 1-6\eps\right)  \tag{$x\eqdef \frac{k}{2r} \in (12\eps,1)$} \\
  &\leq e^{- r D( 1- x \mid\mid 1-6\eps)} \tag{Relative entropy Chernoff bound\footnotemark}
\end{align*}
\footnotetext{Recall that the relative entropy version of the Chernoff bound states that
$
F(m,N,p)\leq e^{-mD\left( \frac{m}{N} \mid\mid p\right) }
$
as long as $0\leq \frac{m}{N} \leq p$.
}
where $D(p\mid\mid q) \eqdef p\ln\frac{p}{q} + (1-p)\ln\frac{1-p}{1-q}$, and we used the fact that $\frac{k}{2}+1 \leq r \leq \frac{k}{24\eps}$. Rewriting slightly the right-hand-side, we obtain
\begin{align*}
  \probaOf{ \sum_{i=1}^{r}  Y_i \geq \frac{k}{2} + 1 }
  &\leq e^{- \frac{k}{2}\Phi(x)} 
\end{align*}
for $\Phi(x) \eqdef \frac{1}{x}\left( (1-x)\ln\frac{1-x}{1-6\eps}+x\ln\frac{x}{6\eps} \right)$. It is not hard to see that $\Phi$ is increasing on $[6\eps,1)$, and since $x\geq 12\eps$ the right-hand-side is at most $e^{- \frac{k}{2}\Phi(12\eps)}$. It then suffices to observe that, for $\eps\leq \frac{1}{12}$, it holds that $\Phi(12\eps) \geq \Phi(0) = \ln 2-\frac{1}{2} > \frac{1}{8}$ to conclude that
\begin{align*}
  \probaOf{ \sum_{i=1}^{r}  Y_i \geq \frac{k}{2} + 1 }
  &\leq e^{- \frac{k}{16}}
\end{align*}
and therefore obtain
 \[
    \probaOf{  f(y_1),\dots,f(y_{r}) \text{ contains at least $(k/2+1)$ zeros} } \leq e^{- \frac{k}{16}}.
    \]
Combining the two, this shows that the probability that $y_1,\dots,y_{r}$ does not reveal a violation for $g_v$  while $g_v$ is $\eps$-far from $k$-monotone is at least $1- e^{-k/16} - e^{-3k/8} > 1/3$ (since $k\geq 8$). By Yao's principle, for any (possibly randomized) non-adaptive algorithm $\Algo$ making at most $k/(24\eps)$ there exists a fixed $v$ such that $g_v$ is $\eps$-far from $k$-monotone yet $\Algo$ rejects $g_v$ with probability less than $2/3$. The desired conclusion follows.
\end{proof}

\subsection{Upper bound for two-sided testers: proof of~\autoref{theo:test:d1:2s:na:ub}}\label{ssec:line:twosided}
In this section, we prove the two-sided non-adaptive upper bound of~\autoref{theo:test:d1:2s:na:ub}, restated below:

\testlinetwosidednaub*

In what follows, we assume that $k > {20}/{\eps}$, as otherwise we can use for instance the $\bigO{k/\eps}$-query (non-adaptive, one-sided) tester of~\autoref{theo:test:d1:os:na:ub} to obtain an $\bigO{1/\eps^2}$ query complexity.

\subsubsection{Testing $k$-monotonicity over $[Ck]$}

We begin by giving a $\poly(C/\eps)$-query tester for $k$-monotonicity over the domain $[Ck]$.  The tester proceeds by reducing to support size estimation and using (a slight variant of) an algorithm of Canonne and Rubinfeld~\cite{CanonneR:14}. Let $f\colon [Ck] \to \{0,1\}$, and suppose $f$ is $s$-monotone but not $(s-1)$-monotone.\marginnote{Terminology: \emph{to have mononicity $s$} K: I removed this definition since it is not used again.}{} Then there is a unique partition of $[Ck]$ into $s+1$ disjoint intervals $I_1,I_2,\ldots,I_{s+1}$ such that $f$ is constant on each interval; note that this constant value alternates in consecutive intervals.
We define a distribution $D_f$ over $[s+1]$ such that $D_f(i) = \abs{I_i}/(Ck)$.

The algorithm of~\cite{CanonneR:14} uses ``dual access'' to $D$; an oracle that provides a random sample from $D$, and an oracle that given an element of $D$, returns the probability mass assigned to this element by $D$. 

\begin{theorem}[{\cite[Theorem 14 (rephrased)]{CanonneR:14}}]
In the access model described above, there exists an algorithm that, on input a threshold $n\in\N^\ast$ and a parameter $\eps > 0$, and given access to a distribution $\D$ (over an arbitrary set) satisfying 
\[ \min_{x\in\supp D} \D(x) \geq \frac{1}{n} \]
estimates the support size $\abs{\supp D}$ up to an additive $\eps n$, with query complexity $\bigO{\frac{1}{\eps^2}}$.
\end{theorem}

 We only have access to $D_f$ through query access to $f$.  One difficulty is that, to access $D_f(i)$, we need to determine where $I_i$ lies in $f$.  For example, finding $D_f(k/2)$ requires finding $I_{k/2}$, which might require a large number of queries to $f$.

We circumvent this by noting that the algorithm does not require knowing the ``label'' of any element in the support of the distribution.
The only access required is being able to randomly sample elements according to $D_f$, and evaluate the probability mass on the sampled points.

\begin{lemma}[Sampling from $D_f$]
Let $i \in [n]$ be chosen uniformly at random, and let $j$ be such that $i \in I_j$.  Then, the distribution of $j$ is exactly $D_f$. 
\end{lemma}

\begin{lemma}[Evaluating $D_f(j)$]
Suppose $I_j = \{a,a+1,\ldots,b\}$.  Given $i$ such that $i \in I_j$, we can find $I_j$ by querying $f(i+1) = f(i+2) = \dots = f(b)$ and $f(b+1) \neq f(b)$, as well as $f(i-1) = f(i-2) = \ldots = f(a)$ and $f(a-1) \neq f(a)$.  The number of queries to $f$ is $b-a+3 = \abs{I_j}+3$.
\end{lemma}

If we straightforwardly use these approaches to emulate the required oracles to estimate the support size of $D_f$, the number of queries is potentially very large.  If we attempt to query $D_f(j)$ where $\abs{I_j} = \Omega(k)$, we will need $\Omega(k)$ queries to $f$.  It will be enough for us to ``cap'' the size of the interval.

\begin{lemma}[Evaluating $D_f(j)$ with a cap]
Given $i$ such that $i \in I_j$, we will query $f$ on every point in $[i - 20 C/\eps, i + 20 C/\eps]$.  If $\abs{I_j} \leq 20 C/\eps$, then $I_j$ will be determined by these queries.  If these queries do not determine $I_j$, we know $\abs{I_j} > 20 C/\eps$.  Beyond querying $i$, this requires $40 C/\eps$ (nonadaptive) queries.
\end{lemma}

\begin{claim}\label{claim:supp:mincase}
If $f$ is $\eps$-far from $k$-monotone, then it is not $(1+\frac{\eps}{4})k$-monotone, and in particular $\abs{ \supp{D_f} } > (1+\frac{\eps}{4})k+1$.
\end{claim}
\begin{proof}
The last part of the statement is immediate from the first, so it suffices to prove the first implication. We show the contrapositive: assuming $f$ is $(1+\frac{\eps}{4})k$-monotone, we will ``fix'' it into a $k$-monotone function by changing at most $\eps n$ points. In what follows, we assume $\frac{\eps k}{4} \geq 1$, as otherwise the statement is trivial (any function that is $\eps$-far from $k$-monotone is \textit{a fortiori} not $k$-monotone).

Let as before $\ell^\ast$ be the minimum integer $\ell$ for which $f$ is $\ell$-monotone: we can assume $k < \ell^\ast \leq (1+\frac{\eps}{4})k$ (as if $\ell^\ast \leq k$ we are done.) 
Consider as above the maximal consecutive monochromatic intervals $I_1,\dots, I_{\ell^\ast}$, and let $i$ be the index of the shortest one. In particular, it must be the case that $\abs{I_i} \leq \frac{n}{\ell^\ast+1}$. Flipping the value of $f$ on $I_i$ therefore has ``cost'' at most $\frac{n}{\ell^\ast+1}$, and the resulting function $f^\prime$ is now exactly $(\ell^\ast-2)$-monotone if $1<i<\ell^\ast$, and  $(\ell^\ast-1)$-monotone if $i\in\{1,\ell^\ast\}$. This means in particular that repeating the above $\frac{\eps}{4} k$ times is enough to obtain a $k$-monotone function, and the total cost is upperbounded by
\[
\sum_{j=0}^{\eps k/4} \frac{n}{\ell^\ast+1-2j} \leq \sum_{j=0}^{\eps k/4} \frac{n}{k+1-2j} 
=  \sum_{j=k(1-\frac{\eps}{2})+1}^{k+1} \frac{n}{j}
\leq n \frac{\frac{\eps}{2}k+1}{(1-\frac{\eps}{4})k+1}
\leq n \frac{\frac{3\eps}{4} k}{(1-\frac{\eps}{4})k}
\]
where for the last inequality (for the numerator) we used that $1\leq \frac{\eps k}{4}$. 
But this last RHS is upperbounded by $\eps n$ (as $\frac{3}{4}x \leq x(1-\frac{1}{4}x)$ for $x\in [0,1]$), showing that 
Therefore, $f$ was $\eps$-close to $k$-monotone to begin with, which is a contradiction.
\end{proof}

\begin{claim}\label{claim:test:d1:2s:na:ub:claim:1}
To $\eps$-test $k$-monotonicity of $f\colon [n] \to \{0,1\}$, it suffices to estimate $\abs{D_f}$ to within $\frac{\eps k}{10}$.
\end{claim}
\begin{proof}
If $f$ is $\eps$-far from $k$-monotone, then $\abs{D_f} > (1 + \frac{\eps}{4})k = k + \frac{\eps}{4}k$, and if $f$ is $k$-monotone, then $\abs{D_f} \leq k + 1$. The fact that $k > 20/\eps$ then allows us to conclude.
\end{proof}

\begin{claim}\label{claim:test:d1:2s:na:ub:claim:2}
There exists a two-sided non-adaptive tester for $k$-monotonicity of functions $f\colon [Ck]\to\{0,1\}$ with query complexity  $\bigO{\frac{C^{3}}{\eps^3}}$.
\end{claim}
\begin{proof}
We use the algorithm of~\cite{CanonneR:14} for estimating support size.  

Inspecting their algorithm, we see that our cap of $20C/\eps$ for interval length (and therefore $20/(\eps k)$ for maximum probability reported) might result in further error of the estimate.  The algorithm interacts with the unknown function by estimating the expected value of $1/D_f(j)$ over random choices of $j$ with respect to $D_f$.  Our cap can only decrease this expectation by at most $(\eps k)/20$.
Indeed, the algorithm works by estimating the quantity $\shortexpect_{x\sim D_f}[ \frac{1}{D_f(x)} \indic{D_f(x) > \tau} ]$, for some suitable parameter $\tau > 0$. By capping the value of  ${1}/{D_f(x)}$ to $20/(\eps k)$, we can therefore only decrease the estimate, and by at most $20/(\eps k)\cdot D_f( \setOfSuchThat{x}{D_f(x) > (\eps k)/20} )  \leq 20/(\eps k)$.

The condition for their algorithm to estimate support size to within $\pm \eps m$ is that all elements in the support have a probability mass of at least $1/m$.  Since each nonempty interval has length at least $1$, we have $\min_j D_f(j) \geq (1/Ck)$.  In order for their algorithm to report an estimate within $\pm \eps k/20$ of support size,
we set $\eps^\prime = (\eps/20C)$ in their algorithm. 

The total error in support size is at most $\eps k/20 + \eps k/20 = \eps k/10$. By~\autoref{claim:test:d1:2s:na:ub:claim:1}, this suffices to test $\eps$-test $k$-monotonicity of $f$.

Using the algorithm of~\cite{CanonneR:14}, we need $O(1/{\eps^\prime}^2) = \bigO{(C/\eps)^2}$ queries to $D_f$.  For every query to $D_f$, we need to make $\bigO{C/\eps}$ queries to $f$, so the overall query complexity is $\bigO{C^{3}/\eps^3}$.
\end{proof}
 
\subsubsection{Reducing $[n] \to \{0,1\}$ to $[Ck] \to \{0,1\}$.}

Now we show how to reduce $\eps$-testing $k$-monotonicity of $f\colon [n] \to \{0,1\}$ to $\eps^\prime$-testing $k$-monotonicity of a function $g\colon [Ck] \to \{0,1\}$ for $C = \poly(1/\eps)$ and $\eps^\prime = \poly(\eps)$, resulting in a $\poly(1/\eps)$-query algorithm for $\eps$-testing $k$-monotonicity.

The first step is (as before) to divide $[n]$ in blocks (disjoint intervals) of size $\frac{\eps n}{4k}$ if $\eps > \frac{8k}{n}$ (again assuming without loss of generality that $\frac{\eps n}{4k}$ is an integer), and blocks of size $1$ otherwise (in which case $n\leq \frac{8k}{\eps}$ and we can directly apply the result of~\autoref{claim:test:d1:2s:na:ub:claim:2}, with $C= n/k \leq 8/\eps$). Let $m=4k/\eps$ be the number of resulting blocks, and define $f_m\colon[n]\to\{0,1\}$ as the \emph{$m$-block-coarsening}\marginnote{Terminology: \emph{$m$-block-coarsening}}\ of $f$: namely, for any $j \in B_i$, we set

\[
f_m(j) = \operatorname{argmax}_{b \in \{0,1\}} \Pr_{k \in B_i}[f_m(k) = b] \tag{majority vote}
\]
Ordering the blocks $B_1,B_2,\ldots,B_m$, we also define $g\colon [m] \to \{0,1\}$ such that $g(i) = \min_{a \in B_i} f_m(a)$.

\begin{lemma}
Suppose $f$ is $k$-monotone.  Then $f$ has at most $k$ non-constant blocks, and $f_m$ is $k$-monotone.
\end{lemma}
\begin{proof}
The function $f$ only changes values $k$ times; for a block to be non-constant, the block must contain a pair of points with a value change.
\end{proof}

We call a block \emph{variable} if the minority points comprise at least an $\eps/100$-fraction of the block; formally, $B$ is variable if $\min_{b \in \{0,1\}} \Pr_{j \in B}[f(j) = b] \geq \eps/100$.

\begin{lemma}\label{lemma:test:d1:2s:na:ub:lemma:2}
Suppose $f$ has $s$ variable blocks.  Then $\dist{f}{f_m} \leq s/m + \eps/100$. 
\end{lemma}

\begin{proof}
We will estimate the error of $f_m$ in computing $f$ on variable blocks and non-variable blocks separately.  Each non-variable block $B$ can contribute error on at most $\eps\abs{B}/100$ points.  Each variable block $B$ can contribute error on at most $\abs{B}=n/m$ points. The total number of errors is at most $\eps n/100 + s(n/m) = n(\eps/100 + s/m)$, yielding the upper bound on $\dist{f}{f_m}$.  
\end{proof}

\begin{lemma}\label{lemma:test:d1:2s:na:ub:lemma:3}
Suppose $f$ is promised to be either (i) $k$-monotone or (ii) such that $f_m$ has more than $\frac{5}{4}k$ variable blocks.  Then we can determine which with $\bigO{\frac{1}{\eps^2}\log \frac{1}{\eps}}$ queries, and probability $9/10$.
\end{lemma}
\begin{proof}
We first note that given any fixed block $B$, it is easy to detect whether it is variable (with probability of failure at most $\delta$) by making $\bigO{\frac{1}{\eps}\log\frac{1}{\delta}}$ uniformly distributed queries in $B$. Doing so, a variable block will be labelled as such with probability at least $1-\delta$, while a constant block will never be marked as variable. (If a block is neither constant nor variable, then any answer will do.)

Letting $s$ denote the number of variable blocks, we then want to non-adaptively distinguish between $s\geq \frac{5}{4}k = \frac{5\eps}{16}m$ and $s\leq k = \frac{\eps}{4}m$ (since if $f$ were $k$-monotone, then $f_m$ had at most $k$ variable blocks). Doing so with probability at least $19/20$ can be done by checking only $q=\bigO{\frac{1}{\eps}}$ blocks chosen uniformly at random: by the above, setting $\delta=\frac{1}{20q}$ all of the $q$ checks will also yield the correct answer with probability no less than $9/10$, so by a union bound we will distinguish (i) and (ii) with probability at least $9/10$. We conclude by observing that all $\bigO{q\cdot \frac{1}{\eps} \log \frac{1}{q}} = \bigO{\frac{1}{\eps^2}\log \frac{1}{\eps}}$ queries are indeed non-adaptive.
\end{proof} 

\begin{claim}
There exists a two-sided non-adaptive tester for $k$-monotonicity of functions $f\colon [n]\to\{0,1\}$ with query complexity $\tildeO{\frac{1}{\eps^7}}$.
\end{claim}

\begin{proof}
We use the estimation/test from the previous lemma as the first part of our tester.  Assuming $f$ passes, we can assume that $f_m$ has less than $\frac{5}{4}k$ variable blocks.  By~\autoref{lemma:test:d1:2s:na:ub:lemma:2}, $\dist{f}{f_m} \leq \frac{5k}{4}/m + \frac{\eps}{100} = \frac{5\eps}{32} + \frac{\eps}{100} \leq \frac{\eps}{3}$. This part takes  $\bigO{\frac{1}{\eps^2}\log\frac{1}{\eps}}$ queries.

Now, we apply the tester of~\autoref{claim:test:d1:2s:na:ub:claim:2} (with probability of success amplified to $9/10$ by standard arguments) to $(\eps/6)$-test $k$-monotonicity of $g\colon [m] \to \{0,1\}$, where $g(i)$ is the constant value of $f_m$ on $B_i$, and $m = {(4k)}/{\eps}$. Let $q$ be the query complexity of the tester, and set $\delta={1}/({10q})$; to query $g(i)$, we randomly query $f$ on $\bigO{\frac{1}{\eps}\log \frac{1}{\delta}}$ points in $B_i$ and take the majority vote.  With probability at least $1-\delta$, we get the correct value of $g(i)$, and by a union bound all $q$ simulated queries have the correct value with probability at least $9/10$.

Therefore, to get a single query to $g$, we use $\bigO{(\log q)/\eps}$ queries.  In the context of our previous section, we have $C = 4/\eps$, so $q =O(C^3/\eps^3) = \bigO{{1}/{\eps^{6}}}$ and the overall query complexity of this part is $\bigO{{(q\log q)}/{\eps}} = \bigO{\frac{1}{\eps^{7}}\log\frac{1}{\eps}}$. This dominates the query complexity of the other part of the tester, from~\autoref{lemma:test:d1:2s:na:ub:lemma:3}, which is $\bigO{\frac{1}{\eps^2}\log\frac{1}{\eps}}$. By a union bound over the part from~\autoref{lemma:test:d1:2s:na:ub:lemma:3}, the simulation of $g$, and the call to the tester of~\autoref{claim:test:d1:2s:na:ub:claim:2}, the algorithm is correct with probability at least $1-3/10 > 2/3$.
\end{proof}

\section{On the grid}\label{sec:grid}
\makeatletter{}We now turn to the grid, and consider $k$-monotonicity of functions defined on $[n]^2$. More specifically, in this section we prove~\autoref{theo:test:d2:2s:a:ub}, giving an adaptive tester for $2$-mononicity with optimal query complexity, before discussing in~\autoref{ssec:grid:kk} possible extensions of these ideas.

\subsection{The case $k=2$}\label{ssec:grid:k2}

\testgridtwosidedaub*
\begin{proof}
At a high-level, the algorithm relies on two key components: the first is the observation that testing $2$-monotonicity of $f\colon[n]^2\to\{0,1\}$ \emph{under some suitable additional assumption on $f$} reduces to (tolerant) testing monotonicity of two \emph{one-dimensional} functions (but with larger range), under the $\lp[1]$ norm. The second is that, given access to an arbitrary $f$, one can efficiently provide query access to some function $g$ which satisifies this additional assumption, and such that $g$ will also be close to $f$ whenever $f$ is truly $2$-monotone.

Combining the two then enables one to test this function $g$ for $2$-monotonicity, and then check whether it is also the case that $f$ and $g$ are sufficiently close. The first step, by the above, can be done efficiently by simulationg query access to $g$, which in turn allows to (with some additional tricks) simulate access to the corresponding one-dimensional functions: and invoke on these two functions the $\lp[1]$-tester of~\cite{BermanRY:14}. (The main challenges there lies in performing this two-level simulation while keeping the number of queries to $f$ low enough; which we achieve by carefully amortizing the queries made overall.)

\paragraph{Details.} We hereafter assume without loss of generality that $f$ is identically $0$ on the bottom and top rows, that is $f(1,j)=f(n,j) = 0$ for all $j\in[n]$. (Indeed, we can ensure this is the case by adding two extra columns, extending the domain of $f$ to $[n+2]\times[n]$: note that if $f$ remains $2$-monotone if it was already, and can only decrease its distance to $2$-monotonicity by $O(1/n)$).\footnote{This will be used in the proof of~\autoref{lemma:test:d2:2s:a:ub:indices:sorted:hamming} and~\autoref{lemma:test:d2:2s:a:ub:indices:sorted:l1}.} For the sake of the proof, we will require the notion of \emph{$2$-column-wise-monotonicity}, defined below:
\begin{definition}\label{def:d2:k:columnwise}
A function $f\colon[n]^2\to\{0,1\}$ is said to be \emph{$2$-column-wise-monotone} if, for every $j\in[n]$, its restriction $f_j\colon [n]\times\{j\}\to\{0,1\}$ is $2$-monotone. Given such a function $f$, we define 
the two sequences $(\hseq{f}_j)_{j\in[n]}$ and $(\lseq{f}_j)_{j\in[n]}$ as the sequence of ``changepoints'' in the columns. More formally, we define
\[
    \lseq{f}_j = \min\setOfSuchThat{ i\in[n] }{ f(i,j) \neq f(1,j) } - 1,
    \qquad
    \hseq{f}_j = \max\setOfSuchThat{ i\in[n] }{ f(i,j) \neq f(n,j) } + 1
\]
for every $j$ such that $f|_{[n]\times\{j\}}$ is not constant; and $\lseq{f}_j=\hseq{f}_j=1$ otherwise. Note that we have $\lseq{f}_j \leq \hseq{f}_j$ for every column $j\in[n]$.
\end{definition}

As it turns out, testing $2$-monotonicity of functions \emph{guaranteed to be $2$-column-wise-monotone} reduces to testing monotonicity of these two specific subsequences:
\begin{lemma}\label{lemma:test:d2:2s:a:ub:indices:sorted:hamming}
Let $f\colon[n]^2\to\{0,1\}$ be $2$-column-wise-monotone. If $f$ is $2$-monotone, then both sequences $(\lseq{f}_j)_{j\in[n]}$ and $(\hseq{f}_j)_{j\in[n]}$ are non-increasing. Moreover, if $f$ is $\eps$-far from $2$-monotone, then at least one of the two sequences is $\frac{\eps}{2}$-far from non-increasing (in Hamming distance).
\end{lemma}
It is possible to refine the above statement to obtain, in the second case, a more precise characterization in terms of the \emph{$\lp[1]$ distance} of these two sequences to monotonicity. For conciseness, in the rest of this section we denote by $\kmon{d}{k}$ the class of $k$-monotone functions on $[n]^d$, and will omit the subscript when $k=1$ (i.e., for monotone functions: $\kmon{d}{}=\kmon{d}{1}$).\marginnote{Terminology: $\kmon{d}{k}$}
\begin{lemma}\label{lemma:test:d2:2s:a:ub:indices:sorted:l1}
Let $f\colon[n]^2\to\{0,1\}$ be $2$-column-wise-monotone. If $f$ is $\eps$-far from $2$-monotone then at least one of the two sequences is $\frac{\eps}{2}$-far from non-increasing (in $\lp[1]$ distance). More precisely:
  \begin{equation}\label{eq:test:d2:2s:a:ub:indices:sorted:l1:alt}
     \dist{f}{\kmon{2}{2}} \leq \lp[1](\hseq{f}, \kmon{1}{})+\lp[1](\lseq{f}, \kmon{1}{})
  \end{equation}
\end{lemma}

We defer the proof of these two lemmata to~\autoref{app:omitted}, and turn to the proof of~\autoref{theo:test:d2:2s:a:ub}. We first describe a non-optimal tester making $O(1/\eps^2)$ queries; before explaining how to modify it in order to amortize the number of queries, to yield the desired $O(1/\eps)$ query complexity.

Suppose we are given query access to an arbitrary function $f\colon[n]^2\to\{0,1\}$. For simplicity, as before we assume without loss of generality that $\frac{\eps n}{16}$ is an integer, and partition each column into $K\eqdef\frac{16}{\eps}$ intervals of size $\frac{\eps n}{16}$, that is partition the domain $[n]\times[n]$ into a grid $[\frac{16}{\eps}]\times[n]$ (each column being divided in $K$ blocks $B_1,\dots,B_K$ of size $\frac{\eps n}{16}$).
 This uniquely defines a function $g\colon[n]^2\to \{0,1,\ast\}$: for any point $x=(i,j)\in[n]^2$, we let $\ell\in[K]$ be the index such that $i\in B_\ell= \{b_\ell,\dots, b_{\ell+1}-1\}$, and set:
\begin{itemize}
  \item $g(x)=f(i,b_\ell))$, if $f(i,b_\ell) = f(i,b_{\ell+1}-1)$;
  \item $g(x)=\unknown$, if $f(i,b_\ell) \neq f(i,b_{\ell+1}-1)$;
\end{itemize}
so that $g$ is constant on any ``block'' $B_{\ell}\times \{i\}$.
  Note that we can provide query access to $g$, at the price of an overhead of 2 queries (to $f$) per query (to $g$).

However, this $g$ may not itself be $2$-column-wise-monotone; for this reason, we will instead work with a ``fixed'' version of $g$ which will by construction be $2$-column-wise-monotone. In more detail, we define $\tilde{g}$ to be the $2$-column-wise-monotone function obtained by the following process.
\begin{itemize}
  \item First, we (arbitrarily) set the values $\ast$ to $0$,  so that $g$ becomes a function $g\colon[n]^2\to \{0,1\}$.
  \item Define $\tilde{g}$ by its restriction on each column: letting $(\hseq{g}_j)_{j\in[n]}$ and $(\lseq{g}_j)_{j\in[n]}$ be as defined in~\autoref{def:d2:k:columnwise} (observing that the quantities are well and uniquely defined even if $g$ is not $2$-column-wise-monotone), set 
  \[
    \tilde{g}(i,j) = \begin{cases}
          g(1,j) &\text{ if } i \leq \lseq{g}_j\\
          1-g(1,j) &\text{ if } \lseq{g}_j < i < \hseq{g}_j\\
          g(n,j) &\text{ if } i \geq \hseq{g}_j
    \end{cases}
  \]
\end{itemize}
From this construction, it is clear that $\tilde{g}$ is $2$-column-wise-monotone, and entirely and deterministically determined by $g$ (and therefore by $f$); moreover we have $\tilde{g}=g$ whenever $g$ is itself $2$-column-wise-monotone. Furthermore, any query to $\tilde{g}$ can straightforwardly be answered by making at most queries $\bigO{{1}/{\eps}}$ to $g$, and hence to $f$.

  \begin{itemize}
    \item If $f$ is $2$-monotone, then $g$ is $\frac{\eps}{8}$-close to $f$ and so is $\tilde{g}$; moreover, $\tilde{g}$ is $2$-monotone as well. Therefore $\dist{f}{\tilde{g}} \leq \frac{\eps}{8}$ and $\dist{\tilde{g}}{\kmon{2}{2}} = 0$.
    \item If $f$ is $\eps$-far from $2$-monotone, then $\tilde{g}$ is either (i) $\frac{\eps}{4}$-far from $f$ or (ii) $\frac{3\eps}{4}$-far from $2$-monotone, since if neither hold then by the triangle inequality $f$ is \eps-close to $2$-monotone. 
  \end{itemize}

\paragraph{The tester (first take).} The algorithm now proceeds as follows:
\begin{enumerate}
  \item simulate query access to $\tilde{g}$ (defined as above) to detect if $\dist{\tilde{g}}{\kmon{2}{2}} \geq \eps$ using~Eq.\eqref{eq:test:d2:2s:a:ub:indices:sorted:l1:alt}, with probability of failure $\frac{1}{6}$. More precisely, test monotonicity in $\lp[1]$ distance of both $\lseq{\tilde{g}}$ and $\hseq{\tilde{g}}$ with parameter $\frac{\eps}{64}$, using the (non-adaptive) algorithm of~\cite{BermanRY:14}; and reject if any of these two tests rejects.
  \begin{itemize}
    \item If $\dist{\tilde{g}}{\kmon{2}{2}} = 0$, then 
    \[
        \lp[1](\hseq{f}, \kmon{1}{}) = \lp[1](\lseq{f}, \kmon{1}{}) = 0
    \]
    by~\autoref{lemma:test:d2:2s:a:ub:indices:sorted:hamming}, and both tests will accept.
    \item If $\dist{\tilde{g}}{\kmon{2}{2}} > \frac{3\eps}{4}$, then
    \[
        \max( \lp[1](\hseq{f}, \kmon{1}{}), \lp[1](\lseq{f}, \kmon{1}{}) ) > \frac{3\eps}{8}
    \]
    and at least one of the two tests rejects.
  \end{itemize} 
  \item simulate query access to $\tilde{g}$ to test whether $\dist{f}{\tilde{g}}\leq \frac{\eps}{8}$ vs. $\dist{f}{\tilde{g}} > \frac{\eps}{4}$, with probability of failure $\frac{1}{6}$;
  \item return \accept if both of the two tests above passed, \reject otherwise.
\end{enumerate}
By a union bound, the tester is then correct with probability at least $\frac{2}{3}$; its query complexity is
\[
    2t\cdot \bigO{\frac{1}{\eps}}+\left( t^\prime\cdot\bigO{\frac{1}{\eps}} + \bigO{\frac{1}{\eps}} \right) 
\]
where $t$ and $t^\prime$ are respectively the cost of simulating query access to $\hseq{g},\lseq{g}$, and to $\tilde{g}$; as the first step, testing in $\lp[1]$ for for functions defined on the line $[n]$, has query complexity $\bigO{1/\eps}$ from~\cite{BermanRY:14}. Taking $t=t^\prime = \bigO{{1}/{\eps}}$ as discussed above then results in a query complexity of $\bigO{{1}/{\eps^2}}$.\medskip

However, as mentioned previously this is not optimal: as we shall see, we can modify this to obtain instead an $\bigO{{1}/{\eps}}$ query complexity. In order to ``amortize'' the overall query complexity, we define the following process that specifies a $2$-column-wise-monotone function $\ring{g}$:
\begin{description}
  \item[Initialization.] Let $j_1,\dots,j_{1/\eps+1}\in[n]$ be the indices defined by $j_\ell = (\ell-1)\cdot \eps n+1$ for $\ell\in[1/\eps]$, and $j_{1/\eps+1}=n$. 
  \begin{itemize}
      \item Obtain all the values of $g$ on the $j_1$-st column $[n]\times\{j_1\}$, at a cost of $O(1/\eps)$ queries to $f$, to find $\hseq{\ring{g}}_{j_1},\lseq{\ring{g}}_{j_1}$. Define $\ring{g}$ on this column accordingly.
      \item Assuming $\ring{g}$ has been defined on the $j_\ell$-th column, define it on the $j_{\ell+1}$-th column: starting at the ``vertical'' positions of the two changepoints $\hseq{\ring{g}}_{j_\ell}, \lseq{\ring{g}}_{j_\ell}$ of the previous column, start querying ``downwards'' the values of $g$ on the $j_{\ell+1}$-th column until candidates values for $\hseq{\ring{g}}_{j_{\ell+1}}, \hseq{\ring{g}}_{j_{\ell+1}}$ consistent with $g$ are found or the bottom of the column is reached (in which case the corresponding changepoint $\hseq{\ring{g}}_{j_{\ell+1}}$ or $\lseq{\ring{g}}_{\ell+1}$ is set to $1$). After this, define $\ring{g}$ on the $j_{\ell+1}$-th column to be consistent with these (at most) two changepoints. 
  \end{itemize}
  Note that the queries made in this step are adaptive, and the (partial) function $\ring{g}$ obtained at the end coincides with $\tilde{g}$ if $f$ (and therefore $\tilde{g}$) is indeed $2$-monotone. This is because in this case, by~\autoref{lemma:test:d2:2s:a:ub:indices:sorted:hamming} the sequences $(\hseq{\tilde{g}}_j)_j,(\lseq{\tilde{g}}_j)_j$ will be non-decreasing, and therefore the process outlined above will result in $\hseq{\tilde{g}}_{j_\ell} = \hseq{\ring{g}}_{j_\ell}$ and $\lseq{\tilde{g}}_{j_\ell} = \lseq{\ring{g}}_{j_\ell}$ for all $1\leq \ell \leq 1/\eps+1$.
  Moreover, it is not difficult to see that the total number of queries made to $f$ in this initialization step will be $\bigO{1/\eps}$: this is because of the fact that we only search for the current changepoint $\hseq{\ring{g}}_{j_\ell}$ (resp. $\lseq{\ring{g}}_{j_\ell}$) starting at the position of the previous one $\hseq{\ring{g}}_{j_{\ell-1}}$ (resp. $\lseq{\ring{g}}_{j_{\ell-1}}$), going downwards only. Since to obtain a changepoint we only query the positions of $g$ (and therefore $f$) at block endpoints, there are in total at most $1/\eps$ positions to query where starting at one and then only going ``down.'' Thus, the number of queries made for each column $j_\ell$ can be written as $\bigO{1}+m_\ell$, where $\sum_{\ell=1}^{1/\eps+1}m_\ell \leq K = \bigO{1/\eps}$. 
  \item[Query time.] When querying the value of $\ring{g}$ on a point $(i,j)$, first let $\ell$ be the index such that $j_\ell \leq j < j_{\ell+1}$. Then define   $\hseq{\ring{g}}_{j}$ (resp. $\lseq{\ring{g}}_{j}$) by querying the value of $g$ on the $j$-th column for all at most $1/\eps$ (block) indices between $\hseq{\ring{g}}_{j_\ell}$ and $\hseq{\ring{g}}_{j_{\ell+1}}$ (resp. $\lseq{\ring{g}}_{j_\ell}$ and $\lseq{\ring{g}}_{j_{\ell+1}}$) to find the corresponding candidate changepoints:
      \begin{align*}
      \lseq{\ring{g}}_j &= \min\setOfSuchThat{ \lseq{\ring{g}}_{j_{\ell+1}} \leq i \leq \lseq{\ring{g}}_{j_{\ell}} }{ g(i,j) \neq g(1,j) } - 1,
    \\
      \hseq{\ring{g}}_j &= \max\setOfSuchThat{ \hseq{\ring{g}}_{j_{\ell+1}} \leq i \leq \hseq{\ring{g}}_{j_{\ell}} }{ g(i,j) \neq g(n,j) } + 1
      \end{align*}
\end{description}
Again, note that after each query the (partial) function $\ring{g}$ obtained so far will coincide with $\tilde{g}$ if $f$ (and therefore $\tilde{g}$) is indeed $2$-monotone; and $\ring{g}$ is uniquely determined by $f$ (and in particular does not depend on the actual queries made nor on their order). Finally, the function $\ring{g}$ thus defined will always by construction be $2$-column-wise-monotone.

The tester previously described can then be slightly modified to simulate access to $\ring{g}$ instead of $\tilde{g}$: by the above discussion, for the completeness case we will have the same guarantees as then $\ring{g}=\tilde{g}$, while the soundness case stays unchanged:
  \begin{itemize}
    \item If $f$ is $2$-monotone, then $\ring{g}=\tilde{g}$ is $\frac{\eps}{8}$-close to $f$; so $\dist{f}{\ring{g}} \leq \frac{\eps}{8}$ 
          and $\dist{\ring{g}}{\kmon{2}{2}} \leq \frac{\eps}{8}$.
    \item If $f$ is $\eps$-far from $2$-monotone, then $\ring{g}$ is either (i) $\frac{\eps}{4}$-far from $f$ or (ii) $\frac{3\eps}{4}$-far from $2$-monotone.
  \end{itemize}
Thus, the analysis of correctness of the tester carries through with this modification; it only remains to bound the query complexity. We will show that the \emph{expected} number of queries made is $\bigO{1/\eps}$; a bound on the worst-case query complexity will then follow from standard arguments,\footnote{Namely, stopping the algorithm and outputting \reject if the number of queries made exceeds $\frac{C}{\eps}$ for some absolute constant $C>0$, found by applying Markov's inequality.} at the price of a constant factor in the~$\bigO{\cdot}$.

To give this bound on the expected number of queries, we first observe that the algorithm from~\cite{BermanRY:14} we rely on in the first stage of the tester is non-adaptive, and moreover all the queries it makes are uniformly distributed (as it works by a reduction, invoking the non-adaptive, one-sided, sample-based monotonicity tester for functions $[n]\to\{0,1\}$). Similarly, all the queries made in the second stage are uniformly distributed as well.

Therefore, the expected number of queries $a_\ell$ made to columns with indices $j_\ell \leq j < j_{\ell+1}$ is the same for each $\ell \in [1/\eps+1]$, namely 
\[
    a_\ell = \frac{q}{1/\eps} = \bigO{1}.
\]
where $q=\bigO{\frac{1}{\eps}}$ is the total number of queries made to $\ring{g}$ (and/or to $\hseq{\ring{g}},\lseq{\ring{g}}$) during the second phase of ``Query time.'' Now, letting $m_\ell = \lseq{\ring{g}}_{j_{\ell}} - \lseq{\ring{g}}_{j_{\ell+1}}$ and $m^\prime_\ell = \hseq{\ring{g}}_{j_{\ell}} - \hseq{\ring{g}}_{j_{\ell+1}}$, we have that the total expected cost of simulating these queries (in terms of queries to $f$) is upperbounded by
\[
    \sum_{\ell=1}^{\frac{1}{\eps}} a_\ell \cdot \left( \bigO{1} + m_\ell + m^\prime_\ell \right) 
    = \bigO{\frac{1}{\eps}} + \bigO{1}\cdot \sum_{\ell=1}^{\frac{1}{\eps}} m_\ell + \bigO{1}\cdot \sum_{\ell=1}^{\frac{1}{\eps}} m^\prime_\ell
    = \bigO{\frac{1}{\eps}}.
\]
Since the total number of queries made to $f$ is the sum of the number of queries made to partly build $\ring{g}$ during the ``Initialization phase'' (which is $\bigO{1/\eps}$ by the foregoing discussion), the number of queries made to simulate access to $\ring{g}$ or $\hseq{\ring{g}},\lseq{\ring{g}}$ during the ``Query time'' (which was just shown to be $\bigO{1/\eps}$ in expectation), and the number of queries directly made to $f$ when testing the distance of $f$ to $\ring{g}$ (which is also $\bigO{1/\eps}$), the expected total number of queries is indeed $\bigO{1/\eps}$, as claimed.
\end{proof}

\subsection{Possible extensions}\label{ssec:grid:kk}
We now discuss two possible extensions of the techniques underlying~\autoref{theo:test:d2:2s:a:ub}, namely to (i) $k$-monotonicity testing of functions over $[n]^2$, for general $k$; and (ii) $2$-monotonicity testing of functions over $[n]^d$, for general $d$. (Note that we do provide a different tester for general $k$ and $d$ in the next section,~\autoref{sec:high:dim}).

\paragraph{Extending to general $k$ (for $d=2$).}
\newest{
A natural direction would be to first to generalize~\autoref{def:d2:k:columnwise} to \emph{$k$-column-wise monotone functions} $f$, defining the $k$ sequences $(\lseq{f}^{(1)}_j)_{j\in[n]},\dots,(\lseq{f}^{(k)}_j)_{j\in[n]}$ of column changepoints with $\lseq{f}^{(1)}_j\leq \dots\leq \lseq{f}^{(k)}_j$ for all $j\in[n]$.
The next step would then be to obtain an analogue of the key lemma of the previous section, \autoref{lemma:test:d2:2s:a:ub:indices:sorted:l1} to this setting. An issue is that it appears necessary to consider now the $\lp[1]$ distance to \emph{$(k-1)$-monotonicity} of these $k$ sequences, instead of monotonicity as before. Thus, taking this route requires to generalize the definition of $k$-monotonicity to real-valued functions, but also to develop $\lp[1]$-testers \emph{for $k$-monotonicity} over the line.

The testing algorithm now follows the same outline as in the previous section, with the same ``amortizing'' idea when invoking this newly obtained $\lp[1]$-tester for $k$-monotonicity in parallel on the $k$ subsequences, each with probability of failure $\delta={1}/{(10k)}$ (for a union bound) and approximation parameter $\eps^\prime = \eps/(10k)$. (Note that some more optimizations may then help further reduce the query complexity, by ``sharing'' the same set of queries between the $k$ instances of the $\lp[1]$-testing algorithm.)
}

\paragraph{Extending to general $d$ (for $k=2$).}
\newest{At a very high-level, the tester of~\autoref{ssec:grid:k2} works by reducing $2$-monotonicity testing of $f\colon[n]^2\to\{0,1\}$ to monotonicity $\lp[1]$-testing of $\lseq{f},\hseq{f}\colon[n]\to[0,1]$. More generally, one can hope to extend this approach to higher dimensions, reducing $2$-monotonicity testing of $f\colon[n]^d\to\{0,1\}$ to monotonicity $\lp[1]$-testing of $\lseq{f},\hseq{f}\colon[n]^{d-1}\to[0,1]$: that is, testing monotonicity (in $\lp[1]$) of the two $(d-1)$-dimensional ``surfaces'' of changepoints. This in turn could be done invoking the $\lp[1]$ non-adaptive tester of~\cite{BermanRY:14} for monotonicity over $[n]^d$, which has query complexity $\tildeO{d/\eps}$: which may lead to a total query complexity of $\poly(d,1/\eps)$, that is \emph{polynomial in the dimension}. We leave this possible extension as an interesting direction for future work.}

\section{On the high-dimensional grid}\label{sec:high:dim}
\makeatletter{}\newcommand{\f}[1]{[#1]}
\newcommand{\fn}{\f{n}}
\newcommand{\fm}{\f{m}}
\newcommand{\calP}{\mathcal{P}}
\newcommand{\calB}{\mathcal{B}}
\newcommand{\calS}{\mathcal{S}}

\algnewcommand{\LineComment}[1]{\State \(\triangleright\) {\sf #1}}

In this section, we give two algorithms for tolerant testing, that is testing whether a function $f\colon \fn^d \to \{0,1\}$ is $\eps_1$-close to $k$-monotone vs. $\eps_2$-far from $k$-monotone, establishing~\autoref{theo:tol:test:d:na:full}. The first has query complexity exponential in the dimension $d$ and is \emph{fully tolerant}, that is works for any setting of $0\leq \eps_1<\eps_2 \leq 1$. The second applies whenever $\eps_2 > 3\eps_1$, and has (incomparable) query complexity exponential in $\tilde{O}(k\sqrt{d}/(\eps_2-3\eps_1)^2)$.  Both of these algorithms can be used for non-tolerant (``regular'') testing by setting $\eps_1 = 0$ and $\eps_2 = \eps$, which implies~\autoref{coro:tol:test:d:na}.

\toltesthighdim*

\noindent As a corollary, this implies~\autoref{coro:tol:test:d:na}, restated below:
\testhighdim*

For convenience, we will view in this part of the paper the set $\fn$ as $\fn = \{0,1,\ldots,n-1\}$.  Assuming that $m$ divides $n$, we let $\calB_{m,n}\colon [n]^d \to [m]^d$ be the mapping such that $\calB_{m,n}(y)_i = \left\lfloor {y_i}/{m} \right\rfloor$ for $1 \leq i \leq m$.  For $x \in \fm^d$, we define the set $\calB^{-1}_{m,n}(x)$ to be the inverse image of $x$.  Specifically, $\calB^{-1}_{m,n}(x)$ is the set of points of the form $m \cdot x + \f{n/m}^d$, with standard definitions for scalar multiplication and coordinate-wise addition.  That is, $\calB^{-1}_{m,n}(x)$ is a ``coset'' of $\f{n/m}^d$ points in $\fn^d$. To keep on with the notations of the other sections, we will call these cosets \emph{blocks}, and will say a function $h\colon[n]^d\to\{0,1\}$ is an \emph{$m$-block function} if it is constant on each block.  Moreover, for clarity of presentation, we will omit the subscripts on $\calB$ and $\calB^{-1}$ whenever they are not necessary.\marginnote{Terminology: \emph{blocks}, \emph{$m$-block function}}

We first establish a lemma that will be useful for the proofs of correctness of both algorithms.
\begin{lemma}
\label{lem:closekmono}
Suppose $f\colon [n]^d \to \{0,1\}$ is $k$-monotone.  Then there is an $m$-block function $h\colon [n]^d \to \{0,1\}$ such that $\dist{f}{h} < kd/m$.
\end{lemma}
\begin{proof}
Fix any $k$-monotone function $f\colon [n]^d \to \{0,1\}$. We partition $[m]^d$ into chains of the form 
\[
C_x = \setOfSuchThat{x + \ell \cdot 1^d }{ \ell \in \N, x \in [m]^d \mathrm{\; and \;} x_i = 0 \mathrm{\; for \; some \; } i }.
\]  There are $m^d - (m-1)^d \leq dm^{d-1}$ of these chains: we will show that $f$ can only be nonconstant on at most $k$ blocks of each chain.

By contradiction, suppose there exists $x\in[m]^d$ such that $f$ is nonconstant on $k+1$ different blocks $\calB^{-1}(z^{(i)})$, where $z^{(1)} \prec z^{(2)} \prec \ldots \prec z^{(k)} \prec z^{(k+1)}$, and each $z^{(i)} \in C_x$.  
By construction, we have $\calB^{-1}(z^{(i)}) \prec \calB^{-1}(z^{(j)})$ for $i < j$.  For each $1\leq i \leq k+1$, there are two points $v^{(i)}_*,v^{(i)}_{**} \in \calB^{-1}(z_i)$ such that $v^{(i)}_* \prec v^{(i)}_{**}$ and $f(v^{(i)}_*) \neq f(v^{(i)}_{**})$. By construction $v^{(1)}_* \prec v^{(1)}_{**} \prec v^{(2)}_* \prec v^{(2)}_{**} \prec v^{(3)}_* \prec v^{(3)}_{**} \prec \ldots \prec v^{(k+1)}_* \prec v^{(k+1)}_{**}$, and there must be at least $k+1$ pairs of consecutive points with differing function values.  Out of these $2k+2$ many points, there is a chain of points $\bar{v}^{(1)} \prec \bar{v}^{(2)} \prec \ldots \prec \bar{v}^{(k+1)}$ where $f(\bar{v}^{(i)}) \neq f(\bar{v}^{(i+1)})$ for $1\leq i \leq k$, which is a violation of the $k$-monotonicity of $f$.

Thus, in each of the $dm^{d-1}$ many chains of blocks, there can only be $k$ nonconstant blocks.  It follows that there are at most $kdm^{d-1}$ nonconstant blocks in total. We now define $h(y)$ to be equal to $f(y)$ if $f$ is constant on $\calB(y)$, and arbitrarily set $h(y) = 0$ otherwise.  Each set $\calB^{-1}(y)$ contains $(n/m)^d = n^d \cdot m^{-d}$ many points, and $f$ is not constant on at most $kdm^{d-1}$ of these.  It follows that $\dist{f}{h} \leq kdm^{d-1} \cdot m^{-d} = kd/m$.
\end{proof}
\subsection{Fully tolerant testing with $O(kd/(\eps_2-\eps_1)))^d$ queries}\label{ssec:high:dim:full:tolerant}

Our first algorithm (\autoref{algo:tester:exp:in:d}) then proceeds by essentially brute-force learning an $m$-block function close to the unknown function, and establishes the first item of~\autoref{theo:tol:test:d:na:full}.

\begin{algorithm}
  \caption{Fully tolerant testing with $O(kd/(\eps_2-\eps_1)))^d$ queries.}\label{algo:tester:exp:in:d}
  \begin{algorithmic}[1]
    \Require Query access to $f\colon \fn^d \to \{0,1\}$, $\eps_2 > \eps_1 \geq 0$, a positive integer $k$
    \State $\alpha \gets (\eps_2 - \eps_1), m \gets \lceil 5kd / \alpha \rceil, t \gets \lceil 25 \ln (6m^d) /(2\alpha^2) \rceil$
    \LineComment{ Define a distribution $D$ over $[m]^d \times \{0,1\}$.}
    \For{$x \in \fm^d$}
       \State Query $f$ on $t$ random points $T_x \subseteq \calB^{-1}(x)$.
       \State $D(x,0) \gets \probaDistrOf{y \in T_x}{ f(y) = 0 }/m^d$
       \State $D(x,1) \gets \probaDistrOf{y \in T_x}{ f(y) = 1 }/m^d$
    \EndFor
    \LineComment{ Define a distribution $D'$ over $[n]^d \times \{0,1\}$ such that $D'(y,b) = D(\calB(y),b) \cdot m^d/n^d$.}
    \If{there exists a $k$-monotone $m$-block function $h$ such that $\probaDistrOf{(y,b) \sim D'}{ h(y) \neq b } \leq \eps_1 + \frac{\alpha}{2}$}
        \Return \accept
    \EndIf \\
    \Return \reject
  \end{algorithmic}
\end{algorithm}

\begin{proposition}\label{prop:tester:exp:in:d}
\autoref{algo:tester:exp:in:d} accepts all functions $\eps_1$-close to $k$-monotone functions, and rejects all functions $\eps_2$-far from $k$-monotone (with probability at least $2/3$).  Its query complexity is $\bigO{ \frac{d}{(\eps_2-\eps_1)^2} \left(\frac{5kd}{\eps_2-\eps_1}+1\right)^d \log \frac{kd}{\eps_2-\eps_1} }$.
\end{proposition}
\begin{proof}
The algorithm first estimates $\probaDistrOf{y \in \calB^{-1}(x)}{f(y) = b}$ for every $x \in \fm^d$ and $b \in \{0,1\}$ to within $\pm \frac{\alpha}{5}$.  We use $t = 25\ln (6 m^d)/2\alpha^2$ points in each block to ensure (by an additive Chernoff bound) that each estimate is correct except with probability at most $m^{-d}/3$.  By a union bound, the probability that all estimates are correct is at least $2/3$, and we hereafter condition on this.  By construction, $\shortexpect_{(x,b) \sim D}[\probaDistrOf{y \in \calB^{-1}(x)}{f(y) \neq b}] = \probaDistrOf{(y,b) \sim D'}{f(y) \neq b} \leq \frac{\alpha}{5}$.  In this probability experiment, the marginal distribution of $D'$ on $y$ is uniform over $\fn^d$.

Let $f^\ast\colon \fn^d \to \{0,1\}$ be a $k$-monotone function minimizing $\probaOf{ f(y) \neq f^*(y) }$.  \autoref{lem:closekmono} ensures that there is a $k$-monotone $m$-block function $h\colon [n]^d \to \{0,1\}$ such that $\dist{ f^* }{ h } < kd/m \leq \alpha/5$.  Let $h^* \colon \fn^d \to \{0,1\}$ be a $k$-monotone $m$-block function minimizing $\dist{ f^* }{ h^\ast }$.

\paragraph*{Completeness.} Suppose $\dist{f}{f^\ast} \leq \eps_1$.  Then by the triangle inequality,
\[
\probaDistrOf{(y,b) \sim D'}{h^*(y) \neq b} 
\leq \probaDistrOf{(y,b) \sim D'}{h^*(y) \neq f^*(y)} + \probaDistrOf{(y,b) \sim D'}{f^*(y) \neq f(y)} + \probaDistrOf{(y,b) \sim D'}{f(y) \neq b} 
\leq \eps_1 + \frac{2\alpha}{5}.
\]
where to bound the first term $\probaDistrOf{(y,b) \sim D'}{h^*(y) \neq f^*(y)}$ by $\dist{ f^* }{ h^\ast } \leq \alpha/5$ we used the fact that the marginal distribution of $y$ is uniform when $(y,b)\sim D^\prime$.
Thus, the algorithm will find a $k$-monotone $m$-block function close to $D$ (without using any queries to $f$) and accept.

\paragraph*{Soundness.} Suppose $\dist{f}{f^\ast}  \geq \eps_2$.  Then by the triangle inequality
\begin{align*}
\Pr_{(y,b) \sim D'}[h(y) \neq b] & \geq \Pr_{(y,b) \sim D'}[h(y) \neq f(y)] - \Pr_{(y,b) \sim D'}[f(y) \neq b] \\ & \geq \Pr_{(y,b) \sim D'}[f^*(y) \neq f(y)] - \Pr_{(y,b) \sim D'}[f(y) \neq b] \\ & \geq \eps_2 - \frac{\alpha}{5}
\end{align*}
for every $k$-monotone $m$-block function $h$.  Since $\eps_2 - 2\alpha/5 \geq \eps_1 + 3\alpha/5$, the algorithm never find a $k$-monotone $m$-block function $h$ with low error with respect to $D$, and the algorithm will reject.

\paragraph*{Query complexity.} The algorithm only makes queries in constructing $D$; the number of queries required is $m^d \cdot t = \bigO{ \frac{d}{\alpha^2} \left(\frac{5kd}{\alpha}+1\right)^d \log \frac{kd}{\alpha} }$.
\end{proof}

\subsection{Tolerant testing \textit{via} agnostic learning}\label{ssec:high:dim:agnostic}
We now present our second algorithm, \autoref{algo:tester:exp:in:rootd}, proving the second item of~\autoref{theo:tol:test:d:na:full}. At its core is the use of an \emph{agnostic learning algorithm} for $k$-monotone functions, which we first describe.\footnote{Recall that an \emph{agnostic learner with excess error $\tau$} for some class of functions \class is an algorithm that, given an unknown distribution $D$, an unknown arbitrary function $f$, and access to random labelled samples $\langle x, f(x)\rangle$ where $x\sim D$, satisfies the following. It outputs a hypothesis function $\hat{h}$ such that $\probaDistrOf{x\sim D}{f(x) \neq \hat{h}(x)} \leq \opt_D + \tau$ with probability at least $2/3$, where $\opt_D = \min_{h\in\class} \probaDistrOf{x\sim D}{f(x) \neq h(x)}$ (i.e., it performs ``almost as well as the best function in $\class$''). }

\begin{proposition}\label{prop:high:dim:agnostic}
There exists an agnostic learning algorithm for $k$-monotone functions over $[r]^d \to \{0,1\}$ with excess error $\tau$ with sample complexity $\exp(\widetilde{O}(k\sqrt{d}/\tau^2)$.
\end{proposition}

\begin{algorithm}
\caption{Multiplicative approximation with $\exp(\widetilde{O}(k\sqrt{d}/((\eps_2-3\eps_1)^2))$ queries}\label{algo:tester:exp:in:rootd}
  \begin{algorithmic}[1]
    \Require Query access to $f\colon \fn^d \to \{0,1\}$, $\eps_2 > 3\eps_1 \geq 0$, a positive integer $k$
    \State $\alpha \gets (\eps_2 - 3\eps_1)$, $m \gets \clg{ 6kd / \eps }, t \gets \clg{ 3d(k+1)/\eps \ln m + \ln 100 }$
    \State Define $D$ to be the distribution over $\fm^d \times \{0,1\}$ such that $D(x,b) = \probaDistrOf{y \in \calB^{-1}(x)}{f(y) = b}$.
    \LineComment{ $\Algo_D(\tau, f)$ denotes the output of an agnostic learner of $k$-monotone functions with respect to $D$, with excess error $\tau$ and probability of failure $1/10$}
    \State\label{algo:tester:exp:in:rootd:learn} $h\colon \fm^d \to \R \gets \Algo_D(\alpha/12, f)$.
    \State\label{algo:tester:exp:in:rootd:estim} Estimate $\probaDistrOf{(x,b) \sim D}{h(x) \neq b}$ to within $\pm {\alpha}/7$ with probability of failure $1/10$, using $O(1/\alpha^2)$ queries.  
    \If{the estimate is more than $\eps_1 + \frac{5\alpha}{12}$} 
        \Return \reject
    \EndIf
    \If{$\dist{h}{\ell} = \probaDistrOf{x \in \fm^d}{h(x) \neq \ell(x)} \leq 2\eps_1 + \frac{5\alpha}{12}$ for some $k$-monotone $m$-block function $\ell$} 
        \Return \accept 
    \Else \; 
        \Return \reject
    \EndIf
  \end{algorithmic}
\end{algorithm}

We will rely on tools from Fourier analysis to prove~\autoref{prop:high:dim:agnostic}. For this reason, it will be convenient in this section to view the range as $\{-1,1\}$ instead of $\{0,1\}$.

\begin{definition}
For a Boolean function $f\colon [r]^d \to \{-1,1\}$, we define 
\[
\infl{i} = 2\probaOf{ [f(x) \neq f(x^{(i)}) }
\]
where $x = (x_1,x_2,\ldots,x_d)$ is a uniformly random string over $[r]^d$, and $x^{(i)} =  (x_1,x_2,\ldots,x_{i-1},x'_i,x_{i+1},\ldots,x_d)$ for $x^\prime$ drawn independently and uniformly from $[r]$.  We also define $\totinf[f] = \sum_{i=1}^d \infl{i}$.
\end{definition}
We first generalize the following result, due to Blais\ \etal., for more general domains:
\begin{proposition}[\cite{BlaisCOST:15}]
Let $f\colon \{0,1\}^d \to \{-1,1\}$ be a $k$-monotone function.  Then $\totinf[f] \leq k\sqrt{d}$.
\end{proposition}
\begin{lemma}[Generalization]\label{lemma:general:influence:kmon}
Let $f\colon [r]^d \to \{-1,1\}$ be a $k$-monotone function.  Then $\totinf[f] \leq k\sqrt{d}$.
\end{lemma}
\begin{proof}
For any two strings $y^0,y^1\in [r]^d$, let $f_{y^0,y^1}\colon \{0,1\}^d \to \{-1,1\}$ be the function obtained by setting $f_{y^0,y^1}(x) = f(y^x)$, where $y^x\in [r]^d$ is defined as
\[
y^x_i = \begin{cases}
      \min\{y^0_i,y^1_i\} & \text{ if }x_i = 0\\
      \max\{y^0_i,y^1_i\} & \text{ if }x_i = 1
      \end{cases}
\]

Since $f$ was a $k$-monotone function, so is $f_{y^0,y^1}$.  Thus $\totinf[f_{y^0,y^1}] \leq k\sqrt{d}$ for every choice of $y^0$ and $y^1$. It is not hard to see that for any fixed $i\in[d]$ the following two processes yield the same distribution over $[r]^d\times[r]^d$:
\begin{itemize}
  \item Draw $z\in[r]^d$, $z^\prime_i\in [r]$ independently and uniformly at random, set $z^\prime \eqdef (z_1,\dots,z_{i-1},z^\prime_i,z_{i+1},\dots,z_d)$, and output $(z,z^\prime)$;
  \item Draw $y^0,y^1\in [r]^d, x\in\{0,1\}^d$ independently and uniformly at random, and output $(y^x, y^{x^{(i)}})$.
\end{itemize}
This implies that 
  \begin{align*}
  \totinf[f] &= \sum_{i=1}^d \infl{i} = \sum_{i=1}^d 2\Pr_{z \in [r]^d}[f(z) \neq f(z^{(i)})] 
  = \sum_{i=1}^d 2\shortexpect_{y^0,y^1 \in [r]^d} \left[ \probaDistrOf{x \in \{0,1\}^d}{ f(y^x) \neq f(y^{x^{(i)}}) } \right]  \\
  &= \shortexpect_{y^0,y^1 \in [r]^d} \left[ \sum_{i=1}^d 2\probaDistrOf{x \in \{0,1\}^d}{ f(y^x) \neq f(y^{x^{(i)}}) } \right] 
  = \shortexpect_{y^0,y^1 \in [r]^d} \left[ \sum_{i=1}^d 2\probaDistrOf{x \in \{0,1\}^d}{ f_{y_0,y_1}(x) \neq f_{y_0,y_1}(x^{(i)}) } \right] \\
  &= \shortexpect_{y^0,y^1}[\totinf[f_{y^0,y^1}]]
  \leq \shortexpect_{y^0,y^1}[ k\sqrt{d} ] = k\sqrt{d}.
  \end{align*}
\end{proof}

For two functions $f,g\colon [r]^d \to \R$, we define the inner product $\dotprod{f}{g} = \shortexpect_x[f(x)g(x)]$, where the expectation is taken with respect to the uniform distribution.  It is known that for functions $f \colon [r]^d \to \R$, there is a ``Fourier basis'' of orthonormal functions $f$.  To construct such a basis, we can take any orthonormal basis $\{\phi_0 \equiv 1,\phi_1,\dots,\phi_{\abs{r}-1}\}$ for functions $f \colon [r] \to \R$.  Given such a basis, a Fourier basis is the collection of functions $\phi_{\alpha}$, where $\alpha \in [r]^d$, and $\phi_{\alpha}(x) = \prod_{i=1}^d \phi_{\alpha_i}(x_i)$.  Then every $f \colon [r]^d \to \R$ has a unique representation $f = \sum_{\alpha \in [r]^d} \widehat{f}(\alpha) \phi_\alpha$, where $\widehat{f}(\alpha) = \dotprod{f}{\phi_\alpha}\in\R$.\medskip

Many Fourier formul\ae{} hold in arbitrary Fourier bases, an important example being Parseval's Identity: $\sum_{\alpha \in [r]^d} \widehat{f}(\alpha)^2 = 1$.  We will use the following property:
\begin{lemma}[{\cite[Proposition 8.23]{AoBF:2014}}]
For $\alpha\in[r]^d$, let $\abs{\alpha}$ denote the number of nonzero coordinates in $\alpha$. Then we have
\[
  \totinf[f] = \sum_{\alpha\in[r]^d} \abs{\alpha} \hat{f}(\alpha)^2.
\]
\end{lemma}
\begin{lemma}\label{lemma:concentration:influence:fourier}
If $\totinf[f] \leq k$, then $\displaystyle\sum_{\alpha : \abs{\alpha} > k/\eps} \widehat{f}(\alpha)^2 \leq \eps$.
\end{lemma}
\begin{proof}
If not, then $\totinf[f] = \sum_{\alpha}\abs{\alpha} \widehat{f}(\alpha)^2 \geq \sum_{\alpha:\abs{\alpha} > k/\eps} \abs{\alpha} \widehat{f}(\alpha)^2 \geq \frac{k}{\eps} \sum_{\alpha : \abs{\alpha} > k/\eps} \widehat{f}(\alpha)^2 > \frac{k}{\eps}\cdot\eps = k$, a contradiction.
\end{proof}

\begin{lemma}
\label{lemma:fourier:three:parts}
Let $p$ be the function $\sum_{\alpha : \abs{\alpha} \leq t} \widehat{f}(\alpha) \phi_\alpha$.  Then
  \begin{enumerate}[(i)]
  \item $\normtwo{p-f}^2=\shortexpect_{x\in[r]^d}[(p(x) - f(x))^2] = \sum_{\alpha : \abs{\alpha} > t} \widehat{f}(\alpha)^2$;
  \item $p$ is expressible as a linear combination of real-valued functions over $[r]^d$, each of which only depends on at most $t$ coordinates;
  \item $p$ is expressible as a degree-$t$ polynomial over the $rd$ indicator functions $\indic{x_i = j}$ for $1 \leq i \leq d$ and $j \in [r]$.
  \end{enumerate}
\end{lemma}

\begin{theorem}[{\cite[Theorem 5]{KalaiKMS:08}}]\label{theo:low:degree:agnostic:learning}
Let $\class$ be a class of Boolean functions over $\domain$ and $\calS$ a collection of real-valued functions over $\domain$ such that for every $f\colon \domain \to \{-1,1\}$ in $\class$, there exists a function $p\colon \domain \to \R$ such that $p$ is expressible as a linear combination of functions from $\calS$ and
$\normtwo{p-f}^2 \leq \tau^2$.  Then there is an agnostic learning algorithm for $\class$ achieving excess error $\tau$ which has sample complexity $\poly(\abs{\calS},1/\tau)$.  
\end{theorem}
Importantly, this algorithm is still successful with inconsistent labelled samples (examples), as long as they come from a distribution on $\domain \times \{-1,1\}$, where the marginal distribution on $\domain$ is uniform.\medskip

Now we put all the pieces together.  To agnostically learn a $k$-monotone function, we simply perform the agnostic learning algorithm of~\cite{KalaiKMS:08} on the distribution $D$ over $[m]^d \times \{-1,1\}$ defined by
\[
    D(x,b) = \probaDistrOf{y\in\calB^{-1}(x)}{f(y)=b}.
\]
To generate a sample $(x,b)$ from $D$, we draw a uniformly random string in $x \in [m]^d$, and $b$ is the result of a query for the value of $f(y)$ for a uniformly random $y \in \calB^{-1}(x)$. From~\autoref{lemma:fourier:three:parts}, we can take $\calS$ to be the set of $(k\sqrt{d}/\tau^2)$-way products of $rd$ indicator functions.  It follows that $\abs{\calS} = \binom{rd}{k\sqrt{d}/\tau^2} = \exp(\widetilde{O}(k\sqrt{d}/\tau^2))$.

\begin{proposition}\label{prop:tester:exp:in:rootd} 
\autoref{algo:tester:exp:in:rootd} accepts all functions $\eps_1$-close to $k$-monotone functions, and rejects all functions $\eps_2$-far from $k$-monotone, when $\eps_2 > 3\eps_1$ (with probability at least $2/3$).  Its query complexity is $\exp(\widetilde{O}(k\sqrt{d}/(\eps_2 - 3\eps_1)^2))$.
\end{proposition}
\begin{proof}
By a union bound, we have that with probability at least $8/10$ both Step~\ref{algo:tester:exp:in:rootd:estim} and Step~\ref{algo:tester:exp:in:rootd:learn} succeed. We hereafter condition on this.

\paragraph*{Completeness.} Suppose $f$ is $\eps_1$-close to $k$-monotone. \autoref{lem:closekmono} and the triangle inequality imply that there is a $k$-monotone $m$-block function $g^\ast$ such that $\dist{f}{g^\ast} \leq \eps_1 + \alpha/6$.  The agnostic learning algorithm thus returns a hypothesis $h$ such that $\dist{f}{h} \leq \eps_1 + \alpha/4$.  The algorithm estimates this closeness to within $\alpha/7$, so the estimate obtained in Step~\ref{algo:tester:exp:in:rootd:estim} is at most $\eps_1 + \eps/4 + \eps/7 < \eps_1 + 5\alpha/12$ and the algorithm does not reject in this step. By the triangle inequality, $h$ is $(2\eps_1 + 5\alpha/12)$-close to $k$-monotone, and the algorithm will accept.  There is no estimation error here, since no queries to $f$ are required.

\paragraph*{Soundness.} Now suppose $f$ is $\eps_2$-far from $k$-monotone, where $\eps_2 = 3\eps_1 + \alpha$ for some $\alpha > 0$.  Suppose the algorithm does not reject when estimating $\dist{f}{h}$, where $h$ is the hypothesis returned by the agnostic learning algorithm.  Then $\dist{f}{h} \leq \eps_1 + {5\alpha}/{12} + \alpha/7 < \eps_1 + 7\alpha/12$.  By the triangle inequality, if $t$ is a $k$-monotone function, $\dist{h}{t} \geq \dist{f}{t} - \dist{f}{h} > \eps_2 - (\eps_1 + 7\alpha/12) = 2\eps_1 + 5\alpha/12$.  The algorithm will thus reject in the final step.

\paragraph*{Query complexity.} The query complexity of the algorithm is dominated by the query complexity of the agnostic learning algorithm, which is $\exp(\widetilde{O}(k\sqrt{d}/\alpha^2)) = \exp(\tildeO{k\sqrt{d}/(\eps_2 - 3\eps_1)^2})$.
\end{proof}

\section{Tolerant testing and applications to $\lp[1]$-testing}\label{sec:application:l1}
\makeatletter{}We now show how our techniques can be applied to solve an open problem on $\lp[1]$ tolerant testing of monotonicity, asked at the Sublinear Algorithms Workshop 2016~\cite{Sublinear:open:70}.

We start by describing a reduction lemma from $\lp[1]$ distance to monotonicity of functions in $[0,1]^{\domain}$ to Hamming distance to monotonicity of functions in $\{0,1\}^{\domain\times[0,1]}$ (that is, ``trading the range for a dimension''). We note that this idea appears in Berman~\etal.\cite[Lemmata 2.1 and 2.3]{BermanRY:14}, although formulated in a slightly different way. For convenience and completeness, we state and prove here the version we shall use.

In what follows, we let $\domain$ be a discrete partially ordered domain equipped with a measure $\mu$,\footnote{We will only require that $(\domain, \mu)$ be a measurable space with finite measure, that is $\mu(\domain)< \infty$, and shall only hereafter concern ourselves with \emph{measurable} functions.} that is a tuple $(\domain, \preceq, \mu)$; and for a set $\range\subseteq\R$ we denote by $\kmon{\domain\to\range}{}\subseteq \range^\domain$ the set of monotone functions from $\domain$ to $\range$.

\begin{definition}[{Analogue of~\cite[Definition 2.1]{BermanRY:14}}]
For a function $f\colon \domain\to [0,1]$, the \emph{threshold function} $T\circ f\colon \domain\times[0,1]\to\{0,1\}$ is defined by
\[
    T\circ f(x,t) = \indic{f(x) \geq 1-t} = 
    \begin{cases}
      1 & \text{ if } f(x) \geq 1-t\\
      0 & \text{ otherwise.}
    \end{cases}
\]
\end{definition}
\noindent The next fact is immediate from this definition:
\begin{fact}\label{fact:monotone:l1:dimension:range}
    For any $f\colon \domain\to [0,1]$, it is the case that for every $x\in\domain$
       \[
        f(x) = \int_{0}^1 T\circ f(x,t) dt.
       \]
   Moreover, $f\in \kmon{\domain\to[0,1]}{}$ if, and only if, $T\circ f\in \kmon{\domain\times[0,1]\to\{0,1\}}{}$.
\end{fact}
We begin by the following characterization, which is immediately obtained from a corresponding theorem of Berman\ \etal.; before stating a slightly modified version that we shall rely upon. For completeness, the proof of the former can be found in~\autoref{app:omitted}.
\begin{proposition}[{Analogue of~\cite[Lemma 2.1]{BermanRY:14}}]\label{prop:dist:monotone:l1:dimension:range}
For any $f\colon \domain\to [0,1]$,
       \[
        \lpdist[1]{ f }{ \kmon{\domain\to[0,1]}{} } = \lpdist[1]{ T\circ f }{ \kmon{\domain\times[0,1]\to\{0,1\}}{} } = \dist{ T\circ f }{ \kmon{\domain\times[0,1]\to\{0,1\}}{} } 
       \]
\end{proposition}

\begin{proposition}[Rounding and Range-Dimension Tradeoff]\label{prop:dist:monotone:l1:dimension:range:coarse}
For any $f\colon \domain\to [0,1]$ and parameter $m \geq 1$, let $R_m\eqdef \{\frac{1}{m}, \frac{2}{m}, \dots, 1\}$.  We define the \emph{$m$-rounding of $f$} as $\Phi_m\circ f\colon\domain\to R_m$ by
\[
    \Phi_m\circ f(x) = \frac{\clg{ m f(x) }}{m},\quad x\in \domain.
\]
Then we have 
\begin{enumerate}[(i)]
  \item\label{prop:dist:monotone:l1:dimension:range:coarse:i} $\abs{ \lpdist[1]{ f }{ \kmon{\domain\to[0,1]}{} } - \lpdist[1]{ \Phi_m\circ f }{ \kmon{\domain\to R_m}{} } } \leq \frac{1}{m}$;
  \item\label{prop:dist:monotone:l1:dimension:range:coarse:ii}  $\lpdist[1]{ \Phi_m\circ f }{ \kmon{\domain\to R_m}{} } = \dist{ T\circ {\Phi_m\circ f} }{ \kmon{\domain\times R_m\to\{0,1\}}{} }$.
\end{enumerate}
\end{proposition}
\begin{proofof}{\autoref{prop:dist:monotone:l1:dimension:range:coarse}}
  Fix any $m\geq 1$. We start the proof of~\autoref{prop:dist:monotone:l1:dimension:range:coarse:i} by the simple observation that if $f\in \kmon{\domain\to[0,1]}{}$, then $\Phi_m\circ f\in \kmon{\domain\to R_m}{}\subseteq \kmon{\domain\to[0,1]}{}$, that is rounding preserves monotonicity; and that $\Phi_m\circ g= g$ for all $g\colon \domain\to R_m$. This, along with the fact that for all $f\colon\domain\to [0,1]$
  \[
      \lpdist[1]{ f }{ \Phi_m\circ f } = \frac{1}{\mu(\domain)}\int_{\domain} \mu(dx) \underbrace{\abs{ \Phi_m\circ f(x)- f(x)}}_{\leq 1/m} \leq \frac{1}{m}
  \]
  implies by the triangle inequality, for any $g\in \kmon{\domain\to R_m}{}\subseteq \kmon{\domain\to[0,1]}{}$, that 
  \[
      \lpdist[1]{ f }{ \kmon{\domain\to[0,1]}{} } \leq \lpdist[1]{ f }{ g } \leq \frac{1}{m} + \lpdist[1]{ g }{ \Phi_m\circ g } + \lpdist[1]{ \Phi_m\circ f }{ \Phi_m\circ g } = \frac{1}{m} + 0 +  \lpdist[1]{ \Phi_m\circ f }{ g }. 
  \]
  Taking $g\in \kmon{\domain\to R_m}{}$ that achieves  $\lpdist[1]{ \Phi_m\circ f }{ g } = \lpdist[1]{ \Phi_m\circ f }{ \kmon{\domain\to R_m}{} }$, we get
  \[
      \lpdist[1]{ f }{ \kmon{\domain\to[0,1]}{} } \leq \frac{1}{m} + \lpdist[1]{ \Phi_m\circ f }{ \kmon{\domain\to R_m}{} }
  .\] For the other direction, we first note that for any two functions $f,g\colon \domain\to[0,1]$, it is the case that $\lpdist[1]{ f }{ g } \geq \lpdist[1]{ \Phi_m\circ f }{ \Phi_m\circ g } - 1/m$ (which is immediate from the definition of the rounding operator), and taking $g$ to be the closest monotone function to $f$ this readily yields 
  \[
  \lpdist[1]{ f }{ \kmon{\domain\to[0,1]}{} } \geq \lpdist[1]{ \Phi_m\circ f }{ \Phi_m\circ g } - \frac{1}{m} \geq \lpdist[1]{ \Phi_m\circ f }{ \kmon{\domain\to R_m }{} } - \frac{1}{m}.
  \]
  
  Finally, the proof of the second part, \autoref{prop:dist:monotone:l1:dimension:range:coarse:i}, is identical to that of~\autoref{prop:dist:monotone:l1:dimension:range}, replacing the Lebesgue measure on $[0,1]$ by the counting measure on $R_m$. (So that integrals over $[0,1]$ become sums over $R_m$, normalized by $\abs{R_m} = m$.)
\end{proofof}

Given~\autoref{prop:dist:monotone:l1:dimension:range:coarse}, it is now easy to apply the results of~\autoref{sec:high:dim} to obtain a tolerant $\lp[1]$ tester for monotonicity of functions $f\colon[n]^d \to [0,1]$. Indeed, given parameters $0<\eps_1<\eps_2$, one can set the rounding parameter $m$ to $\clg{4/(\eps_2-\eps_1)}$; and from query access to $f\colon[n]^d \to [0,1]$, simulate query access to $\Phi_m\circ f$ and therefore to $g\eqdef T\circ \Phi_m\circ f\colon [n]^d\times R_m\to \{0,1\}$. By~\autoref{prop:dist:monotone:l1:dimension:range:coarse} and our choice of $m$, in order to distinguish
\[
\lpdist[1]{ f }{ \kmon{\domain\to[0,1]}{} } \leq \eps_1 \quad\text{ vs. }\quad \lpdist[1]{ f }{ \kmon{\domain\to[0,1]}{} } \geq \eps_2
\]
it is enough to distinguish
\[
\dist{ g }{ \kmon{\domain\times R_m\to\{0,1\}}{} } \leq \eps_1 + \frac{1}{m} \quad\text{ vs. }\quad \dist{ g }{ \kmon{\domain\times R_m\to\{0,1\}}{} } \geq \eps_2 - \frac{1}{m}.
\]
By our choice of $m$, we also have $\left(\eps_2 - \frac{1}{m}\right)-\left(\eps_1 + \frac{1}{m}\right) \geq \frac{\eps_2-\eps_1}{2}$.

The last step is to observe that one can view equivalently $g$ as a function $g\colon [n]^d\times [m]\to \{0,1\}$; by~\autoref{prop:dist:monotone:l1:dimension:range:coarse} and our choice of $m$, so that the algorithms of~\autoref{coro:tol:test:d:na} apply.

\toltestlone*

\paragraph{Acknowledgments.} We would like to thank Eric Blais for helpful remarks on an earlier version of this paper, and an anonymous reviewer for very detailed and insightful comments.

\bibliographystyle{alpha}
\bibliography{references} 

\clearpage
\appendix
\section{Previous work on monotonicity testing}\label{app:overview}
\makeatletter{}In this appendix, we summarize the state-of-the-art on monotonicity testing. We observe that this question has been considered for functions over various domains (e.g. hypergrids, hypercubes and general posets) and ranges (notably Boolean range $\B$ and unbounded range $\N$); as hypergrids and hypercubes are arguably the domains that have received the most attention in the literature, we will in this overview restrict ourselves on work on these, and refer 
readers to \cite{FLNRRS:02,BhattacharyyaGJRW:09} for other various posets.  We will also focus on the Boolean range $\B$, which is most relevant to our work, and briefly mention the best known results (which are also tight) for unbounded range $\N$. In the end of this section, we include known results for \emph{tolerant} monotonicity testing.  

Before we go over those results, we recall some notation: namely, testers can make adaptive (a.) or non-adaptive (n.a.) queries and have $1$-sided (1.s.) or $2$-sided (2.s.) error. The best one could hope for would then be to obtain $1$-sided non-adaptive upper bound, complemented with $2$-sided adaptive lower bounds. We note all testers included in below except tolerant testers are $1$-sided (almost all of them are non-adaptive) algorithms.

 \paragraph{Hypercubes with Boolean Range.} The problem of monotonicity testing is introduced by Goldreich~\etal.~\cite{GGLRS00} for functions $f\colon\B^d\to\B$. \cite{GGLRS00} present
 a simple ``edge tester'' with query complexity $O(d/\eps)$.  A tester with $O(d^{7/8}/\eps^{3/2})$ queries,  the first improvement in terms of the dependence on $d$ and the first to ``break'' the linear barrier,
was presented by Chakrabarty and Seshadhri~\cite{CS:13b}, further improved to $\tilde{O}(d^{5/6}/\eps^4)$ by Chen~\etal.~\cite{ChenST:14}. 
Recently, a $\tilde{O}(\sqrt{d}/\eps^2)$ upper bound was established by Khot~\etal.~\cite{KhotMS:15}.  All these upper bounds are obtained for $1$-sided, non-adaptive testers.

For $1$-sided non-adaptive testers, Fischer~\etal.~\cite{FLNRRS:02} showed an $\Omega(\sqrt{d})$ lower bound.  For $2$-sided non-adaptive testers,
Chen~\etal.~\cite{ChenST:14} obtained  an $\tilde{\Omega}(d^{1/5})$ lower bound, further improved by Chen~\etal.~\cite{ChenDST:15}) 
to $\Omega(d^{1/2-c})$ (for any constant $c>0$).   All these lower bounds applying to non-adaptive testers, they only imply an $\Omega(\log d)$ lower bound 
for adaptive ones.  Recently, Belovs and Blais~\cite{BelovsB:15} showed an
$\Omega(d^{1/4})$ lower bound for $2$-sided adaptive testers, i.e. an exponential improvement over the previous bounds. All mentioned lower bounds hold 
for constant $\eps>0$, and are summarized in~\autoref{table:app:monotonicity:cube-boolean}.

\begin{table}[H]\centering
  \begin{tabular}{|l|c|c|}
  \hline
      Domain & Upper bound &  Lower bound \\
         \hline
   $\B^d$ & $\bigO{\sqrt{d}}$ 1.s.-n.a.~\cite{KhotMS:15} &  $\bigOmega{d^{1/2}}$ 1.s.-n.a.~\cite{FLNRRS:02}\\
   &    &  $\bigOmega{d^{1/2-o(1)}}$ 2.s.-n.a.~\cite{ChenDST:15}\\
                                                        &    &  $\tildeOmega{d^{1/4}}$ 2.s.-a.~\cite{BelovsB:15}\\
  \hline
  \end{tabular}
  \caption{\label{table:app:monotonicity:cube-boolean}Testing monotonicity of a function $f\colon\B^d\to\B$}
\end{table}

\paragraph{Hypergrids with Boolean Range.} 
  We remark that most known previous upper bounds for testing monotonicity over hypergrids are for unbounded range, which we will be the focus of the next section. Instead, we only mention here the case of Boolean range, giving in each setting the current best known results. For testing monotonicity over the line with Boolean range (i.e. $d=1$ case),  both a 1-sided non-adaptive $O(1/\eps)$ upper bound and a 2-sided adaptive $\Omega(1/\eps)$ lower bound are known (both of them being folklore). 
   For $d=2$, Berman \etal.~\cite{BermanRY:14} showed a tight bound of $\Theta((\log{1/\eps})/\eps)$ for 1-sided non-adaptive testers.  
   Interestingly,  they also prove that ``adaptivity'' helps in the $d=2$ case: that is, they establish a 1-sided tight adaptive $O(1/\eps)$ upper bound
  which beats $\Omega(\log{1/\eps})/\eps)$ lower bound for $1$-sided non-adaptive testers.   
    For general $d$, Berman \etal.~\cite{BermanRY:14} give both a 1-sided non-adaptive tester with query complexity $O(\frac{d}{\eps}{\log{\frac{d}{\eps}}})$, and a 1-sided adaptive tester with query complexity  $\bigO{d2^d\log^{d-1}{\frac{1}{\eps}}+\frac{d^2\log{d}}{\eps}}$.  
The best known results can be found in~\autoref{table:grids-boolean}.

\begin{table}[H]\centering
  \begin{tabular}{|l|c|c|}
  \hline
      Domain & Upper bound &  Lower bound \\
         \hline
   $[n]$ & $\bigO{\frac{1}{\eps}}$ 1.s.-n.a.  & $\bigOmega{\frac{1}{\eps}}$ 2.s.-a.\\
   $[n]^2$ & $\bigO{\frac{1}{\eps}\log{\frac{1}{\eps}}}$ 1.s.-n.a.~\cite{BermanRY:14} & $\bigOmega{\frac{1}{\eps}\log{\frac{1}{\eps}}}$ 1.s.-n.a.~\cite{BermanRY:14}\\
                &  $\bigO{\frac{1}{\eps}}$  1.s.-a.~\cite{BermanRY:14} & $\bigOmega{\frac{1}{\eps}}$ 2.s.-a.\\
   $[n]^d$ & $\bigO{\frac{d}{\eps}\log{\frac{d}{\eps}}}$ 1.s.-n.a.~\cite{BermanRY:14} & $\bigOmega{\frac{1}{\eps}\log{\frac{1}{\eps}}}$ 1.s.-n.a.~\cite{BermanRY:14}\\
                &  $\bigO{d2^d\log^{d-1}{\frac{1}{\eps}}+\frac{d^2\log{d}}{\eps}}$  1.s.-a.~\cite{BermanRY:14}& $\bigOmega{\frac{1}{\eps}}$ 2.s.-a.\\
  \hline
  \end{tabular}
  \caption{\label{table:grids-boolean} Testing monotonicity of a function $f\colon[n]^d \to\B$}
\end{table}

 \paragraph{Unbounded Range.} For unbounded range, tight upper and lower bounds are known for both hypergrid and hypercube domains. 
Chakrabarty and Seshadhri~\cite{CS:13} describe a $1$-sided  non-adaptive tester with $O(d\log{n}/\eps)$ queries for the hypergrid $[n]^d$. 
Later, they show that $O(d\log{n}/\eps)$ is essentially optimal even for $2$-sided adaptive tester~\cite{CS:13a}. For the hypercube,
\cite{CS:13} give a $1$-sided  non-adaptive tester making $O(n/\eps)$ queries, and a matching $2$-sided adaptive lower bound  
is proved by Joshua Brody (mentioned as private communication in~\cite{CS:13a}). We refer readers to~\cite{CS:13,CS:13a} for overviews on previous results for testing monotonicity over the hypercube and hypergrid with unbounded range.  
The best known results are summarized in~\autoref{table:unbounded}. 
\begin{table}[H]\centering
  \begin{tabular}{|l|c|c|}
  \hline
      Domain & Upper bound &  Lower bound \\
         \hline
   $\B^d$ & $\bigO{\frac{d}{\eps}}$ 1.s.-n.a.~\cite{CS:13} &  $\bigOmega{\frac{d}{\eps}}$ 2.s.-a.~\cite{CS:13a}\\
    $[n]^d$ & $\bigO{\frac{d\log{n}}{\eps}}$ 1.s.-n.a.~\cite{CS:13}&  $\bigOmega{\frac{d\log{n}}{\eps} - \frac{1}{\eps}\log{\frac{1}{\eps}}}$ 2.s.-a.~\cite{CS:13a}\\
  \hline
  \end{tabular}
  \caption{\label{table:unbounded} Testing monotonicity of a function $f\colon D\to\N$}
\end{table}

\paragraph{Tolerant Testing.} To the best of our knowledge, prior to our work tolerant testers for monotonicity for Boolean functions over the hypergrid were only known for dimension $d\in\{1,2\}$.  Specifically, an $O(\frac{\eps_2}{(\eps_2-\eps_1)^2})$-query upper bound is known for $d=1$, while an $\tilde{O}(\frac{1}{(\eps_2-\eps_1)^4})$-query one is known for $d=2$~\cite{BermanRY:14,FattalR:10}.
 
\begin{table}[H]\centering
  \begin{tabular}{|l|c|c|}
  \hline
      Domain & Upper bound &  Lower bound \\
         \hline
   $[n]$ & $O(\frac{\eps_2}{(\eps_2-\eps_1)^2})$ \cite{BermanRY:14}  2.s.-a.  & ? \\
   $[n]^2$ & $\tilde{O}(\frac{1}{(\eps_2-\eps_1)^4})$ \cite{FattalR:10} 2.s.-n.a. &  ?\\
  \hline
  \end{tabular}
  \caption{\label{table:grids-tolerant} Tolerant testing monotonicity of a function $f\colon[n]^d\to\B$}
\end{table}

\section{Structural results}\label{app:poset:structural}
\makeatletter{}
In this section, we will prove that the distance to $k$-monotonicity of a Boolean function $f$ can be expressed in a combinatorial way -- which does not require measuring the distance between $f$ and the closest $k$-monotone function to $f$. We will prove this for a general finite poset domain buiding up on the ideas of~\cite{FLNRRS:02}. In the rest of this section, we denote by $\poset = (V, \preceq)$ an arbitrary poset, the underlying domain of the function.

\begin{definition} \label{def:forbidden-bit-pattern}
We define the \emph{forbidden pattern} $K_{10}$ as the sequence of alternating bits $K_{10} = (b_1, b_2, \cdots, b_k, b_{k+1})$ of length $(k+1)$ , where $b_1 = 1$ and all the bits in the sequence alternate, i.e., $b_i \neq b_{i+1}\ \forall i \in [k]$.
\end{definition}

\noindent A function $f\colon \poset \to \{0,1\}$ is $k$-monotone only if it avoids $K_{10}$. That is, for any $x_1 \prec x_2 \prec \ldots \prec x_{k+1} \in \poset$ we have $f(x_i) \neq K_{10}(i)$ for some $i \in [k+1]$. 

Using insights from the literature on monotonicity testing, we show that functions far from $k$-monotonicity have a large matching of ``violated hyperedges'' in the ``violation hypergraph'' which we define shortly. Let us recall the definition of ``violation graph''which has been extremely useful with monotonicity testing as seen in~\cite{ErgunKKRV00, PRR:06, HalevyK08, AilonCCL04, FLNRRS:02}.

\begin{definition}[Violation graph]
Given a function $f\colon \poset \to \{0,1\}$, the \emph{violation graph of $f$} is defined as $G_{\rm viol}(f) = ( \poset, E(G_{\rm viol}) )$ where $(x,y) \in E(G_{\rm viol})$ if $x,y \in \poset$ form a monotonicity violating pair in $f$ -- that is $x \preceq y$ but $f(x) > f(y)$. 
\end{definition}
The following theorem about violation graphs has been extremely useful in monotonicity testing literature.
\begin{theorem} \label{theo:mono:big-violated-matching}
Let $f\colon \poset \to \{0,1\}$ be a function that is $\eps$-far from monotone. Then, there exists a matching of edges in the violation graph for $f$ of size at least $\eps \abs{\poset}/2$.
\end{theorem}

\noindent Now let us define a generalization of this concept, the violation hypergraph.

\begin{definition}[Violation hypergraph]
 Given a function $f\colon \poset \to \{0,1\}$, the \emph{violation hypergraph of $f$} is $H_{\rm viol}(f) = ( \poset, E(H_{\rm viol}) )$ where $(x_1, x_2, \cdots, x_k) \in E(H_{\rm viol})$ if the ordered $(k+1)$-tuple $x_1 < x_2 < \ldots < x_{k+1}$ (which is a $(k+1)$-uniform hyperedge) forms a violation to $k$-monotonicity in $f$.
 \end{definition}

\noindent Now, we state the main theorem that we intend to prove in this section. This theorem offers an alternate characterization of distance to $k$-monotonicity that we seek. We recall that a set of edges forms a matching in a hypergraph if any pair of hyperedges is pairwise disjoint.

\begin{theorem} \label{theo:k-mono:big-violated-matching}
Let $f\colon \poset \to \{0,1\}$ be function that is $\eps$-far from $k$-monotone. Then, there exists a matching of $(k+1)$-uniform hyperedges of size at least $\frac{\eps \abs{\poset}}{k+1}$ in the violation hypergraph.
\end{theorem}

To prove this theorem we first exploit the key notion of \emph{extendability} which we define below. Later we will show that $k$-monotone functions are extendable.

\begin{definition}[Extendability]\label{extendable}
 A property of Boolean functions is said to be \emph{extendable} over a poset domain if the following holds for any set $X \subseteq \poset$: given a function $f \colon X \to \{0,1\}$ which has the property (on $X$), it is possible to define a function $g \colon \poset \to \{0,1\}$ such that $g(x) = f(x), \forall x \in X$ and $g$ has the property.
\end{definition}

\noindent In other words, a property is extendable if for any subset $X \subseteq \poset$, given a function defined over the set $X$ which respects the property, it is possible to fill in values outside $X$ such that the new function obtained continues to respect the property. Next, we show that $k$-monotonicity is an extendable property:
\begin{lemma}\label{lemma:k-mono-extendable}
$k$-monotonicity is an extendable property.
\end{lemma}

\begin{proofof} {\autoref{lemma:k-mono-extendable}}
Consider $X \subseteq \poset$ and some function $f\colon X \to \{0,1\}$ which is $k$-monotone over $X$. That is, for any $x_1 < x_2 < \ldots < x_{k+1} \in X$ there exists $i \in [k+1]$ such that $f(x_i) \neq K_{10}[i]$. Take a minimal point $v \in \poset \setminus X$. That is, for any other point $v' \in \poset \setminus X$ either $v \leq v'$ or $v$ and $v'$ are not comparable. We will use the following result:
\begin{claim}\label{claim:extend-by-one}
 There exists a function 
$
    g\colon X \cup \{v\} \to \{0,1\}
$
 such that $g(x) = f(x)$ for all $x \in X$, and $g$ respects $k$-monotonicity over its domain. 
 \end{claim}
Before proving this claim, we show how it implies~\autoref{lemma:k-mono-extendable}. Namely, starting with any function $f\colon X\to\{0,1\}$ which is $k$-monotone on its domain $X$, we just keep applying the~\autoref{claim:extend-by-one} inductively until we get a function defined over the entire poset which respects $k$-monotonicity.
\begin{proofof}{\autoref{claim:extend-by-one}}

We will show this by contradiction. Suppose there is no good assignment available for $g(v)$, that is that both the choices $g(v) = 0$ and $g(v) = 1$ lead to a violation to $k$-monotonicity in $g$. Consider the choice $g(v) = 0$. Since this results in a violation to $k$-monotonicity, we know that there is a path $P_0 = (x_1 \prec x_2 \prec \ldots \prec x_{k+1})$ which is a violation to $k$-monotonicity. It is clear that $v \in P_0$; let $i$ be such that $x_i=v$. Similarly, there is path $P_1 = (y_1 \prec y_2 \prec \ldots \prec y_{k+1})$ corresponding to $g(v) = 1$ which also contains the forbidden pattern, and some $j$ such that $y_j = v$. And thus, $g(x_t) = g(y_t),$ for all $t \in [k+1]$ (as both of the paths indexed by $x$ and $y$ form a violation to $k$-monotonicity).
\noindent By the above discussion $2$ paths, $P_0$ and $P_1$, meet at $v$. We will see that one of the two paths 
\[
    P'_0 = (x_1 \prec x_2 \prec \ldots < x_{i-1} \prec y_i \prec y_{i+1} \prec \ldots \prec y_{k+1})
\]
or 
\[
    P'_1 = (y_1 \prec y_2 \prec \ldots < y_{j-1} \prec x_j \prec x_{j+1} \prec \ldots \prec x_{k+1})
\] is already a violation to $k$-monotonicity in $f$. 
To see this, let us begin by recalling that we let $v$ be the $i^{th}$ vertex on $P_0$ and the $j^{th}$ vertex on $P_1$. Now it is clear that $i \neq j$. Without loss of generality, suppose $i < j$. In this case, the evaluations of $f$ along path $P'_1$ form the forbidden pattern. This is because the function values alternate along the segment $(y_1 \prec y_2 \prec \ldots \prec y_{j-1})$. Also, the function values alternate along the segment $(x_i \prec x_{i+1} \prec x_{i+2} \prec \ldots \prec x_{k+1})$. And finally note that since $f(y_{j-1}) \neq f(y_j)$ and $f(y_j) = f(x_j)$ we get that $f(y_{j-1}) \neq f(x_j)$ as well. So, the path $P'_1$ indeed contains a violation to $k$-monotonicity as claimed. The other case, $i > j$, is analogous. Hence the claim follows.
\end{proofof}
\end{proofof}

In the next lemma, we show that there is a nice characterization of distance to $k$-monotonicity in terms of the size of the \emph{minimum vertex cover} of the violation hypergraph.\footnote{Recall that a vertex cover in a hypergraph is just a set of vertices such that every hyperedge contains at least one of the vertices from this set.}

\begin{lemma} \label{lemma:vtx-cover-dist-to-prop}
Let $\mathcal{M}_k$ denote the set of $k$-monotone functions over the poset $\poset$, and $f\colon \poset \to \{0,1\}$. Then $\dist{f}{\mathcal{M}_k} = \eps_f$ if, and only if, the size of the minimum vertex cover in $H_{\rm viol}(f)$ is $\eps_f \abs{\poset}$.
\end{lemma}
\begin{proofof} {\autoref{lemma:vtx-cover-dist-to-prop}}
We establish separately the two inequalities.
\begin{claim} \label{claim:forward-dir}
$\dist{f}{\mathcal{M}_k} \geq \abs{VC_{min}(H_{\rm viol})}$
\end{claim}
\begin{proof}
\noindent Suppose the distance to $k$-monotonicity is $\eps_f$, and let $g$ be a $k$-monotone achieving it, so that $\dist{f}{g}=\eps_f$. Define $X = \setOfSuchThat{x \in \poset }{ f(x) \neq g(x) }$ (thus, $\abs{X} = \eps_f \abs{\poset}$). Let us consider the violation hypergraph for $f$ given as $H_{\rm viol}(f) = (\poset, E(H_{\rm viol}))$. Now, delete vertices in $X$ and the hyperedges containing any vertex $v \in X$ from this hypergraph. Because $X$ is the smallest set of vertices changing values at which gives a $k$-monotone function, it follows that every hyperedge in $E(H_{\rm viol})$ must contain a vertex in $X$. Thus, $X$ indeed forms a vertex cover in $H_{\rm viol}(f)$.
\end{proof}

\begin{claim} \label{claim:backward-dir}
$\dist{f}{\mathcal{M}_k} \leq \abs{ VC_{min}(H_{\rm viol}) }$
\end{claim}
\begin{proof}
\noindent Suppose the minimum vertex cover in the violation hypergraph has size $\eps_f \abs{\poset}$. We will show that the distance of the function to $k$-monotonicity is $\eps_f$. To see this, let $C \subseteq \poset$ be the smallest vertex cover of the violation hypergraph of the said size. Observe that deleting $C$ from $H_{\rm viol}(f)$ removes all the hyperedges, and therefore that the function $f$ restricted to $X = \poset \setminus C$ is $k$-monotone. And by the extendability of $k$-monotone functions established in~\autoref{lemma:k-mono-extendable}, it follows that the function can be extended to the rest of the domain (by providing values in the cover $C$) such that it keeps respecting $k$-monotonicity. Thus, the distance to $k$-monotonicity is at most $\abs{C}/\abs{\poset}$.
\end{proof}

\end{proofof}

Having characterized distance to $k$-monotonicity as the size of the smallest vertex cover in the violation graph, we are ready to establish~\autoref{theo:k-mono:big-violated-matching}. To do so, we will require the following standard fact:
\begin{fact} \label{claim:max-matching-vtx-cover}
Let $G = (V, E)$ be a $t$-uniform hypergraph. Let $M$ be the maximum matching in $G$. Then, $\abs{M} \leq VC_{min}(G) \leq t\abs{M}$
\end{fact}
\begin{proofof}{\autoref{theo:k-mono:big-violated-matching}}
By~\autoref{lemma:vtx-cover-dist-to-prop}, we know that the violation hypergraph, $H_{\rm viol}(f)$ has a minimum vertex cover of size at least $\eps_f \abs{\poset}$. And by~\autoref{claim:max-matching-vtx-cover}, it is seen that it contains a matching of $t$-uniform hyperedges of size at least $\frac{\eps_f}{t}\abs{\poset}$.
\end{proofof}

\section{Omitted proofs}\label{app:omitted}
\makeatletter{}  \begin{proofof}{\autoref{lemma:test:d2:2s:a:ub:indices:sorted:hamming}}
  The first part of the theorem is straightforward (by contrapositive, if at least one of the two sequences is not non-increasing then we can find a violation of $2$-monotonicity). We thus turn to the second part, and show the contrapositive; for this purpose, we require the following result:  
  \begin{claim}\label{claim:test:d2:2s:a:ub:indices:sorted:hamming:proof}
    If $f\colon[n]^2\to\{0,1\}$ is a $2$-column-wise monotone function such that (i) $f(1,j)=f(n,j)=0$ for all $j\in[n]$ and (ii) both $(\lseq{f}_j)_{j\in[n]}, (\lseq{h}_j)_{j\in[n]}\subseteq[n]$ are non-increasing, then $f$ is $2$-monotone.
  \end{claim}
  \begin{proof}
      By contradiction, suppose there exists a $2$-column-wise monotone function $f$ satisfying (i) and (ii), which is not $2$-monotone. This last point implies there exists a triple of comparable elements $x=(i_x,j_x) \prec y=(i_y,j_y) \prec z=(i_z,j_z)$ constituting a violation, i.e. such that $(f(x),f(y),f(z))=(1,0,1)$. Moreover, since (i) holds we must have $1 < i_x \leq i_y \leq i_z < n$; more precisely, 
       $1\leq \lseq{f}_{j_x} < i_x \leq i_y \leq i_z < \hseq{f}_{j_z} \leq n$. As $x\prec y\prec z$, we have $j_x\leq j_y\leq j_z$, which by the non-increasing assumption (ii) implies that
       $\lseq{f}_{j_x} \geq \lseq{f}_{j_y}$ and $\hseq{f}_{j_y} \geq \hseq{f}_{j_z}$. But this is not possible, as altogether this leads to $\lseq{f}_{j_y} < i_y < \hseq{f}_{j_y}$, i.e. $f(y) = 1$.
  \end{proof}
  
  Assume both sequences $(\lseq{f}_j)_{j\in[n]}, (\lseq{h}_j)_{j\in[n]}\subseteq[n]$ are $\frac{\eps}{2}$-close to non-increasing, and let $L,H\subset[n]$ (respectively) be the set of indices where the two sequences need to be changed in order to become non-increasing. By assumption, $\abs{L},\abs{H}\leq \frac{\eps n}{2}$, so $\abs{L\cup H} \leq \eps n$. But to ``fix'' a value of $(\lseq{f}_j)_{j\in[n]}$ or $(\hseq{f}_j)_{j\in[n]}$ requires to change the values of the function $f$ inside a single column -- and this can be done preserving its $2$-column-wise-monotonicity, so that changing the value of $f$ on at most $n$ points is enough. It follows that making both $(\lseq{f}_j)_{j\in[n]}$ and $(\hseq{f}_j)_{j\in[n]}$ non-increasing requires to change $f$ on at most $\eps n^2$ points, and with~\autoref{claim:test:d2:2s:a:ub:indices:sorted:hamming:proof} this results in a function which is $2$-monotone. Thus, $f$ is $\eps$-close to $2$-monotone.
  \end{proofof}

\begin{proofof}{\autoref{lemma:test:d2:2s:a:ub:indices:sorted:l1}}
Recall that we aim at establishing the following:
  \begin{equation}\label{eq:test:d2:2s:a:ub:indices:sorted:l1:alt:restated}
     \dist{f}{\kmon{2}{2}} \leq \lp[1](\hseq{f}, \kmon{1}{})+\lp[1](\lseq{f}, \kmon{1}{})
  \end{equation}
  For notational convenience, we will view in this proof the sequences $(\lseq{f})_j,(\hseq{f})_j)$ as functions $\lseq{f},\hseq{f}\colon[n]\to[n]$. Let $\ell,h\colon[n]\to[n]$ (for ``low'' and ''high,'' respectively) be monotone functions achieving $\lp[1](\lseq{f}, \kmon{1}{})$ and $\lp[1](\hseq{f}, \kmon{1}{})$, respectively.
  \begin{itemize}
    \item As $\lseq{f}(j) \leq \hseq{f}(j)$ for all $j\in[n]$, we will assume $\ell(j)\leq h(j)$ for all $j$. Otherwise, one can consider instead the functions $\ell^\prime = \min(\ell,h)$ and $h^\prime = \max(\ell,h)$: both will still be monotone (non-increasing), and by construction 
      \[
      \abs{\ell^\prime(j)-\lseq{f}(j)}+\abs{h^\prime(j)-\hseq{f}(j)} \leq \abs{\ell(j)-\lseq{f}(j)}+\abs{h(j)-\hseq{f}(j)}
      \] for all $j\in [n]$, so that
      $\lp[1](\hseq{f}, \ell^\prime)+\lp[1](\lseq{f}, h^\prime) \leq \lp[1](\hseq{f}, \ell)+\lp[1](\lseq{f}, h)$.
    \item From $\ell$ and $h$, we can define a $2$-column-wise monotone function $g\colon[n]^2\to[n]$ such that $\lseq{g}=\ell$ and $\hseq{g}=h$: that is,
      \[
            g(i,j) = \begin{cases}
                0 & \text{ if } i \geq h(j) \\
                1 & \text{ if } \ell(j) < i < h(j) \\
                0 & \text{ if } i \leq \ell(j)
            \end{cases}
      \]
      for $(i,j)\in[n]^2$.
  \end{itemize}
  It is clear that $g$ is $2$-column-wise monotone with $g(1,j)=g(n,j)=0$ for all $j\in[n]$; since by construction $\lseq{g},\hseq{g}$ are non-decreasing, we can invoke~\autoref{claim:test:d2:2s:a:ub:indices:sorted:hamming:proof} to conclude $g$ is $2$-monotone. It remains to bound the distance between $f$ and $g$: writing $\Delta_j\in\{0,\dots,n\}$ for the number of points on which $f$ and $g$ differ in the $j$-th column, we have
  \begin{align*}
    \dist{f}{ \kmon{2}{2} } \leq \dist{f}{g} &= \frac{1}{n^2} \sum_{j=1}^n \Delta_j \leq \frac{1}{n^2} \sum_{j=1}^n \left( \abs{\ell(j)-\lseq{f}(j)}+\abs{h(j)-\hseq{f}(j)} \right) \\
    &= \frac{1}{n^2} \sum_{j=1}^n \abs{\ell(j)-\lseq{f}(j)}+ \frac{1}{n^2} \sum_{j=1}^n \abs{h(j)-\hseq{f}(j)}
    = \lpdist[1]{\lseq{f}}{\ell} + \lpdist[1]{\hseq{f}}{h} \\
    &\leq  \lp[1](\lseq{f}, \kmon{1}{}) + \lp[1](\hseq{f}, \kmon{1}{})
  \end{align*}
  which concludes the proof.
\end{proofof}

\begin{proofof}{\autoref{prop:dist:monotone:l1:dimension:range}}
We write $\nu\eqdef \mu\times\operatorname{Leb}_{[0,1]}$ for the product measure on $\domain\times[0,1]$ induced by $\mu$ and the Lebesgue measure on $[0,1]$; so that $\nu(\domain\times[0,1]) = \mu(\domain)\cdot 1 = \mu(\domain)$.

For any fixed $t\in[0,1]$, let $g_t\in \kmon{\domain\to\{0,1\}}{}$ be any function achieving $\lpdist[1]{T\circ f(\cdot,t)}{g_t}=\lpdist[1]{ T\circ f(\cdot,t) }{  \kmon{\domain\to\{0,1\}}{} }$, and define $g\in[0,1]^\domain$ by $g^\prime(x) = \int_0^1 dt g_t(x)$ for all $x\in\domain$: note that $g$ is then  monotone by construction.\footnote{Additionally, since we restrict ourselves to finite $\domain$, there are only finitely many distinct functions $T\circ f(\cdot,t)$ (for $t\in[0,1]$, and therefore only finitely many distinct functions $g_t$.} Moreover, choose $h\in\kmon{\domain\times[0,1]\to\{0,1\}}{}$ as a function achieving $\lpdist[1]{T\circ f}{h}=\lpdist[1]{ T\circ f }{  \kmon{\domain\times[0,1]\to\{0,1\}}{} }$.
Then we have 
\begin{align*}
    \lpdist[1]{f}{ \kmon{\domain\to[0,1]}{} }  &\leq \lpdist[1]{f}{g^\prime} 
    = \frac{1}{\mu(\domain)}\int_{\domain} \mu(dx) \left\lvert \int_{0}^1 dt ( T\circ f(x,t) - g_t(x) ) \right\rvert \\
    &\leq \frac{1}{\mu(\domain)}\int_{\domain} \mu(dx) \int_{0}^1 dt \left\lvert  T\circ f(x,t) - g_t(x) \right\rvert \\
    &= \int_{0}^1 dt \left( \frac{1}{\mu(\domain)}\int_{\domain} \mu(dx)\left\lvert  T\circ f(x,t) - g_t(x) \right\rvert\right) 
    = \int_{0}^1 dt \lpdist[1]{T\circ f(\cdot,t)}{g_t} \\
    &\leq \int_{0}^1 dt \lpdist[1]{T\circ f(\cdot,t)}{h(\cdot,t)} 
    = \int_{0}^1 dt \left( \frac{1}{\mu(\domain)}\int_{\domain} \mu(dx)\left\lvert  T\circ f(x,t) - h(x,t) \right\rvert\right) \\
    &= \frac{1}{\nu(\domain\times [0,1])} \int_{\domain\times[0,1]} \nu(dx,dt) \left\lvert  T\circ f(x,t) - h(x,t) \right\rvert \\
    &= \lpdist[1]{T\circ f}{h} = \lpdist[1]{ T\circ f }{  \kmon{\domain\times[0,1]\to\{0,1\}}{} }
\end{align*}
where we applied~\autoref{fact:monotone:l1:dimension:range} (and the definition of $g^\prime = \int_0^1 g_t$) for the first equality, and for the third inequality the fact that $h$ induces (for every fixed $t\in[0,1]$) a monotone function $h(\cdot,t)\in \kmon{\domain\to\{0,1\}}{}$: so that $\lpdist[1]{T\circ f(\cdot,t)}{g_t} \leq \lpdist[1]{T\circ f(\cdot,t)}{h(\cdot,t)}$ for all $t$.\medskip

For the other direction of the inequality, fix any $f\colon \domain\to [0,1]$, and let $g\in\kmon{\domain\to[0,1]}{}$ be (any) function achieving $\lpdist[1]{f}{g}=\lpdist[1]{f}{ \kmon{\domain\to[0,1]}{} }$.
We can write, unrolling the definitions,
\begin{align*}
    \lpdist[1]{f}{ \kmon{\domain\to[0,1]}{} } &= \frac{1}{\mu(\domain)}\int_{\domain}  \mu(dx) \lvert f(x) - g(x) \rvert \\
    &= \frac{1}{\mu(\domain)}\int_{\domain} \mu(dx) \left\lvert \int_{0}^1 dt ( T\circ f(x,t) - T\circ g(x,t) ) \right\rvert \\
    &= \frac{1}{\mu(\domain)}\int_{\domain} \mu(dx) \left\lvert \int_{0}^1 dt ( T\circ f(x,t) - T\circ g(x,t) ) \right\rvert \\
    &= \frac{1}{\mu(\domain)}\int_{\domain} \mu(dx) \Big( \int_{0}^1 dt ( T\circ f(x,t) - T\circ g(x,t) )\indic{f(x) > g(x)} \\
    &\qquad+ ( T\circ g(x,t) - T\circ f(x,t) )\indic{g(x) > f(x)} \Big) \\
    &= \frac{1}{\mu(\domain)}\int_{\domain} \int_{0}^1 dt \mu(dx) \Big(( T\circ f(x,t) - T\circ g(x,t) )\indic{f(x) > g(x)} \\
    &\qquad+ ( T\circ g(x,t) - T\circ f(x,t) )\indic{g(x) > f(x)} \Big) \\
    &= \frac{1}{\nu(\domain\times[0,1])}\int_{\domain\times[0,1]} \nu(dx,dt) \left\lvert T\circ f(x,t) - T\circ g(x,t) \right\rvert 
    = \lpdist[1]{ T\circ f }{  T\circ g  } \\
    &\geq \lpdist[1]{ T\circ f }{  \kmon{\domain\times[0,1]\to\{0,1\}}{} }
\end{align*}
where we applied~\autoref{fact:monotone:l1:dimension:range} for the second equality, the definition of $\lp[1]$ distance for the second-to-last; and to handle the absolute values we used the fact that $\lvert a - b\rvert = (a-b)\indic{a > b}+ (b-a)\indic{a > b}$, along with the observation that $T\circ f(x,t) > T\circ g(x,t)$ can only hold if $f(x)> g(x)$. Finally, we have $\lpdist[1]{ T\circ f }{  T\circ g } \geq \lpdist[1]{ T\circ f }{  \kmon{\domain\times[0,1]\to\{0,1\}}{} }$ since $T\circ g\in \kmon{\domain\times[0,1]\to\{0,1\}}{}$, yielding the desired claim.

Finally, the fact that $\lpdist[1]{ T\circ f }{ \kmon{\domain\times[0,1]\to\{0,1\}}{} } = \dist{ T\circ f }{ \kmon{\domain\times[0,1]\to\{0,1\}}{} }$ is immediate from the Boolean range, as $\lvert a - b\rvert = \indic{a\neq b}$ for any $a.b\in\{0,1\}$.
\end{proofof}

\end{document}